\documentclass[cleveref,thm-restate]{lipics-v2021}

\pdfoutput=1
\hideLIPIcs
\nolinenumbers

\graphicspath{{./figures}}

\usepackage{todonotes}
\usepackage{multirow}

\usepackage{xspace}
\usepackage{mleftright}

\crefname{property}{Property}{Properties}

\newtheorem{property}[theorem]{Property}

\newcommand{\f}{Fr\'echet\xspace}
\newcommand{\dF}{\mathop{d_\mathrm{F}}}

\newcommand{\eps}{\varepsilon\xspace}
\newcommand{\from}{\colon}

\newcommand{\R}{\ensuremath{\mathbb{R}}}
\newcommand{\D}{\ensuremath{\mathcal{D}}}
\newcommand{\F}{\ensuremath{\mathcal{F}}}
\newcommand{\RM}{\ensuremath{\mathit{RM}}}
\newcommand{\BL}{\ensuremath{\mathit{BL}}}
\newcommand{\TR}{\ensuremath{\mathit{TR}}}

\newcommand{\calT}{\ensuremath{\mathcal{T}}}

\newcommand{\calS}{\ensuremath{\mathcal{S}}}

\newcommand{\calR}{\ensuremath{\mathcal{R}}}

\newcommand{\enumit}[1]{{\textcolor{lipicsGray}{\sffamily\bfseries\upshape\mathversion{bold}#1}}}
\newcommand{\lab}[1]{{\sffamily\textbf{#1}}}

\title{A Strongly-Subquadratic \texorpdfstring{$(3+\varepsilon)$}{(3+ε)}-Approximation for the Fr\'echet Distance for Paths in Metric Spaces}
\titlerunning{A Strongly-Subquadratic \texorpdfstring{$(3+\varepsilon)$}{(3+ε)}-Approximation for the Fr\'echet Distance}

\author
{Thijs van der Horst}
{Department of Information and Computing Sciences, Utrecht University, the Netherlands
\and
Department of Mathematics and Computer Science, TU Eindhoven, the Netherlands}
{t.w.j.vanderhorst@uu.nl}
{https://orcid.org/0009-0002-6987-4489}
{}

\author
{Tim Ophelders}
{Department of Mathematics and Computer Science, TU Eindhoven, the Netherlands}
{t.a.e.ophelders@tue.nl}
{https://orcid.org/0000-0002-9570-024X}
{partially supported by the Dutch Research Council (NWO) under project no.\ VI.Vidi.243.236.}

\authorrunning{T. van der Horst and T. Ophelders}

\Copyright{Thijs van der Horst, and Tim Ophelders}



\begin{document}
\maketitle

\begin{abstract}
    The Fr\'echet distance is a well-studied distance measure for paths in a metric space.
    It is mostly studied for paths in $d$-dimensional Euclidean space.
    Here, computing the Fr\'echet distance between two polylines takes time roughly quadratic in the number of vertices.
    Assuming the strong exponential time hypothesis (SETH), it cannot be approximated to within a factor less than $3$ in strongly-subquadratic time.
    Recently, it was shown that for any $\varepsilon>0$, there exists a randomized algorithm that can compute a $(7+\varepsilon)$-approximation in strongly-subquadratic expected time [Cheng, Huang, and Zhang; STOC'25].
    For polylines with $n$ and $m$ vertices in a Euclidean space of constant dimension, where $n \geq m$, their algorithm takes $O(nm^{0.99} \log(n/\varepsilon))$ time in expectation.

    We present a deterministic approximation algorithm that significantly improves upon the approximation factor and running time.
    Specifically, our algorithm computes a $(3+\varepsilon)$-approximation in $O(nm^{2/3} \log n \cdot \log (\frac{1}{\varepsilon} \log n))$ time.
    Our algorithm nearly matches the conditional lower bound on the approximation factor implied by SETH.
    For polylines in $\mathbb{R}$, we present a $3$-approximation algorithm that runs in $O(nm^{2/3} \log^{5/3} n)$ time, and exactly matches the conditional lower bound.
    
    For our results, we introduce a general strongly-subquadratic time $3$-approximate decision algorithm.
    This algorithm makes no assumptions on the ambient metric space, and relies only on standard assumptions on the so-called free space of the input paths.
    Under some mild assumptions, our decision algorithm leads to a $(3+\varepsilon)$-approximation algorithm in general metric spaces.
    These assumptions hold automatically for polylines in any metric space $(\mathbb{R}^d, L_p)$ with $p \geq 1$.
\end{abstract}

% \newpage
\section{Introduction}

    The \f distance is a well-studied distance measure for paths in a metric space, used to determine the similarity of two paths.
    Measuring similarity is an important task in for example matching time series in data bases~\cite{liu13time_series_data_base}, protein alignment~\cite{jiang08protein_alignment}, and trajectory analysis~\cite{su2020survey}.
    The \f distance is more discriminative than, e.g., the Hausdorff distance, as it takes into account the intrinsic ordering of the points on a path.
    
    Alt and Godau~\cite{alt95continuous_frechet} were the first to study the computability of the \f distance.
    They showed that one can compute the \f distance between two polylines in $\R^d$, under some $L_p$ norm with $p \geq 1$, in $O(nm \log n)$ time, where $n$ and $m$ are the number of vertices of the polylines and $n \geq m$.
    Since their result, only marginal improvements have been made~\cite{buchin17continuous_frechet,cheng25subquadratic_frechet}, shaving off mere polylogarithmic factors.
    The lack of significant improvements can be explained by the widely-believed strong exponential time hypothesis (SETH).
    Bringmann~\cite{bringmann14hardness} proved that assuming SETH, the \f distance cannot be computed in strongly-subquadratic time; i.e., $O((nm)^{1-\delta})$ time for any $\delta > 0$.
    Moreover, for ${n = \tilde\Theta(m)}$, Buchin, Ophelders, and Speckmann~\cite{buchin19seth_says} showed that the \f distance cannot be approximated better than within a factor $3$ in strongly-subquadratic time, even for $d = 1$.
    
    The conditional lower bound raises the question of how well the \f distance can be approximated in strongly-subquadratic time.
    For polylines in $(\R^d, L_2)$, several approximation algorithms are known~\cite{bringmann16approx_discrete_frechet,cheng25constant_factor_frechet,vanderHorst24faster_Frechet}.
    Until recently, the best approximation algorithms presented a trade-off between running time and approximation factor.
    In particular, for an $O(n^{2-\delta})$-time algorithm, a polynomial approximation factor of $n^\delta$ was needed~\cite{vanderHorst24faster_Frechet}.
    
    This left the question of whether a constant factor approximation in strongly-subquadratic time is possible.
    Cheng, Huang, and Zhang~\cite{cheng25constant_factor_frechet} came very close to answering this question in~2025, by giving a randomized $(7+\eps)$-approximation algorithm with an expected running time of $O(nm^{0.99} \log(n/\eps))$ for any $\eps > 0$, assuming the dimension $d$ is constant.
    Recently, Blank~\cite{blank26imbalanced_frechet} gave the first $(3+\eps)$-approximation algorithm that runs in strongly-subquadratic time when the input complexities are imbalanced; that is, when $m = n^c$ for some $c \in (0, 1)$.
    Specifically, they presented an algorithm with running time $O((n+m^2) \log(n/\eps))$.
    Even more recently, independently and concurrently with our work, Liu and Wang~\cite{liu26constant_factor_frechet} improved the approximation factor by~\cite{cheng25constant_factor_frechet} to $5+\eps$, and the expected running time to $\tilde O(nm^{8/9})$ in the continuous setting, and $\tilde O(nm^{4/5})$ in the discrete setting, assuming $d$ and $\eps$ are constant.\textbf{}
    
    Algorithms for computing the \f distance have also been studied in other settings.
    For example, Rote~\cite{rote07piecewise_smooth} and Conradi, Driemel, and Kolbe~\cite{conradi25piecewise_smooth} study the setting where the input paths are not polylines, but (well-behaved) piecewise-smooth curves.
    Even in this generalization, we can compute the \f distance in near-quadratic time.
    Cook and Wenk~\cite{cook10geodesic_frechet} consider polylines in a more restricted metric space, namely that of a simple $k$-gon endowed with the geodesic distance.
    In this setting the \f distance can also be computed in near-quadratic time, namely $O(k + n^2 \log^2 n)$.
    Maheshwari and Yi~\cite{maheshwari05frechet_on_polyhedron} consider polylines that lie on a convex polyhedron, and give a polynomial-time algorithm for this case.

    \begin{table}[b]
    \centering
    \caption{\label{tab:results_decider}
        $O(1)$-approximate deciders for the discrete and continuous \f distance between polylines of complexity $n$ and $m\leq n$, assuming freespace cells can be queried in $T_O=O(1)$ time.
    \vspace{-1ex}}%
    \begin{tabular}{l|l|l|c|c|l}
        Running time & Space use & \shortstack{Apx.\\factor} & Constraints & Excels at & Reference\\
        \hline\rule{0pt}{2.6ex}%
        $O(n+m^2)$              & $O(n)$                 & 2        & $\R^1$                     & $2$-approximation        & \cite{blank26imbalanced_frechet}\\
        $O(n+m^2)$              & $O(n)$                 & 3        &                            & $m^{11/9}\leq n\log n$   & \cref{thm:decider_imbalanced} or \cite{blank26imbalanced_frechet}\!\!\!\\
        $\tilde O(n+m^{5/3})$   & $\tilde O(n+m^{5/3})$  & 7        &                            & $7$-approximation & \cref{cor:decider_imbalanced_7apx}\\
        \hline\rule{0pt}{2.6ex}%
        $O(nm^{0.99})$ expected &                        & $7+\eps$ &                            & - & \cite{cheng25constant_factor_frechet}\\
        $\tilde O(nm^{8/9})$ expected &                  & $5+\eps$ & continuous                 & - & \cite{liu26constant_factor_frechet}\\
        $\tilde O(nm^{4/5})$ expected &                  & $5$      & discrete                   & - & \cite{liu26constant_factor_frechet}\\
        $\tilde O(nm^{2/3})$    & $\tilde O(nm^{2/3})$   & 3        &                            & time-efficiency          & \cref{thm:approx_decider}\\
        $\tilde O(nm^{1-\eps})$ & $\tilde O(nm^{2\eps})$ & 3        & $0\leq\eps\leq\frac{1}{3}$ & $m^{1+\eps}\geq n\log n$ & \cref{thm:approx_decider_tau_DS_during_iteration}\\
        $\tilde O(nm^{7/9})$    & $\tilde O(n)$          & 3        &                            & space-efficiency         & \cref{thm:approx_decider_both_DS_during_iteration}
    \end{tabular}
    \end{table}
    
    \begin{table}
    \centering
    \caption{\label{tab:results}
        Our results for approximating the continuous or discrete \f distance between polylines with $n$ and $m\leq n$ vertices.
    \vspace{-1ex}}
    \begin{tabular}{l|l|l|l}
        {Ambient space} & {Metric} & \shortstack{Apx.\\factor} & {Running time}\\
        \hline\rule{0pt}{2.6ex}%
        $\R^{d = O(\log n)}$ & $L_p$ for $p \geq 1$ & $3+\eps$ & $O(nm^{2/3} \cdot \log^{4/3} n \cdot \log(\frac{1}{\eps} \log n))$\\
        $\R^{d = \Omega(\log n)}$ & $L_p$ for $p \geq 1$ & $3+\eps$ & $O(nm^{2/3} \cdot (d^3 + d^2 \log^2 n)^{1/3} \cdot \log(\frac{1}{\eps} \log n))$\\
        $\R$ & absolute difference & $3$ & $O(nm^{2/3} \cdot \log^2 n)$\\
        simple $k$-gon & geodesic Euclidean & $3+\eps$ & $O(nm^{2/3} \cdot \log^{4/3} n \cdot \log(\frac{1}{\eps} \log n) + k)$
    \end{tabular}
    \end{table}
        
    \subparagraph*{Results.}
    We present a new algorithm for approximating the \f distance between two paths in a metric space $(M, d)$.
    Moreover, our algorithm approximates the \f distance between \emph{hybrid paths}, which generalize both continuous paths and discrete sequences of points.
    Our algorithm makes no assumptions on the metric space.
    Instead, it relies only on standard assumptions on the so-called free space of the input paths.
    These assumptions hold in particular for some commonly used settings, shown in \cref{tab:results}.

    In $(\R^d, L_2)$, our algorithm improves upon the algorithms by Cheng, Huang, and Zhang~\cite{cheng25constant_factor_frechet} and by Liu and Wang~\cite{liu26constant_factor_frechet} on several fronts:
    it is deterministic, it has a better approximation factor, and its worst-case (rather than expected) running time beats their expected running times.
    When $m = \Omega(n^{3/4} \log^{1/2} n)$, our algorithm also improves upon the algorithm by Blank~\cite{blank26imbalanced_frechet}.
    In one-dimensional real space, we improve the approximation factor further to $3$, which under SETH is the best approximation factor possible in strongly-subquadratic time when $m=\tilde\Theta(n)$.
    Our algorithm also leads to the first strongly-subquadratic constant-factor approximation for paths in a simple polygon under the geodesic distance.
    
    Our main contribution is a $3$-approximate decider, which we present in \cref{sec:decision}.
    \Cref{tab:results_decider} summarizes the computational complexity.
    On a high level, our approximate decider resembles that by Cheng, Huang, and Zhang.
    We briefly discuss some main differences in \cref{sec:discussion}.
    Our decider applies to a very general class of inputs. 
    Our input assumptions are outlined in \cref{sec:input_assumptions}.
    We turn our $3$-approximate decider into an approximation algorithm for the \f distance through a simple search algorithm, at a slight increase in approximation factor.
    \Cref{sec:optimization} details this step, as well as our results for the settings shown in \cref{tab:results}.

\section{Generalizing the Fr\'echet distance}
    The Fr\'echet distance has been studied extensively for continuous and discrete paths.
    Generalizing, we define a \emph{hybrid $n$-path} as a function $\sigma \from X \to M$ where $X \subseteq [0, 1]$ is a union of $n$ interior-disjoint closed intervals $X_1, \dots, X_n$.
    Hybrid paths simultaneously generalize polylines (which are hybrid paths) and discrete paths (by parameterizing each point by a singleton interval).
    Let $\sigma[x, x']$ denote the restriction of $\sigma$ to the domain~$X \cap [x, x']$.
    We call $\sigma[x,x']$ a \emph{hybrid subpath} of $\sigma$.
    The \emph{pieces} of $\sigma$ are the restrictions of $\sigma$ to the domains $X \cap X_i$.
    Each endpoint of each $X_i$ is called a \emph{breakpoint} of $\sigma$.
        
    In this section, we generalize the Fr\'echet distance to hybrid paths.
    This generalization simplifies the exposition of our algorithms.
    Algorithms for computing the Fr\'echet distance can often easily be adapted to work for hybrid paths.
    For example, the decision algorithm of Alt and Godau~\cite{alt95continuous_frechet} readily extends to hybrid paths, see \cref{app:exact_decision}.

    \subparagraph*{Parameter space and matchings.}
        To define the \f distance between hybrid paths $\sigma \from X \to M$ and $\tau \from Y \to M$ in a common metric space $(M, d)$, we first introduce the concept of a matching between $\sigma$ and~$\tau$.
        We denote by $\D(\sigma, \tau) = X \times Y$ the product of the parameter spaces, see \cref{fig:parameter_space}.
        The point $(x,y)\in\D(\sigma, \tau)$ corresponds to the pair of points $(\sigma(x),\tau(y))$ on $\sigma$ and~$\tau$.
        The parameter space $\D(\sigma, \tau)$ is a collection of interior-disjoint axis-aligned rectangles.
        We refer to each such rectangle as a \emph{cell}.
        These cells are arranged in a grid, which we refer to as the \emph{parameter space grid}.
        The sides of the cells are the \emph{edges} of the grid.
        
        \begin{figure}
            \centering
            \includegraphics{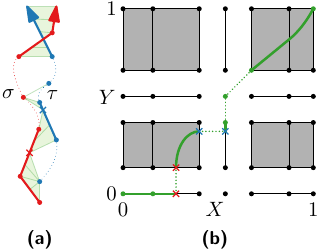}
            \caption{\lab{(a)} Two hybrid paths $\sigma$ and $\tau$ and \lab{(b)} the parameter space $\D(\sigma,\tau)$ and a matching.}
            \label{fig:parameter_space}
        \end{figure}

        A \emph{matching} between $\sigma$ and $\tau$ is a pair $(f,g)$ of non-decreasing surjections ${f \from [0, 1] \to X}$ and ${g \from [0,1] \to Y}$.
        We treat $(f,g)$ as the bimonotone function ${t\mapsto(f(t),g(t))\from[0,1]\to\D(\sigma,\tau)}$.
        We say that the pair $(f(t),g(t))$ is \emph{matched} by $(f,g)$, and require that matched pairs form a compact subset of $\D(\sigma,\tau)$.
        In particular, within each visited cell of $\D(\sigma, \tau)$, the matched pairs form a bimonotone path homeomorphic to a closed interval.
        By surjectivity, the first path starts at the bottom-left corner of $\D(\sigma, \tau)$, and the last path ends at the top-right corner.
        Moreover, if one were to glue adjacent cell-boundaries of $\D(\sigma,\tau)$, the path would become a single continuous bimonotone path.
        The \emph{cost} of a matching $(f,g)$ between $\sigma$ and $\tau$ is the maximum distance in $M$ between matched points:
        \[
            \max_{t\in[0,1]} d(\sigma(f(t)),\tau(g(t))).
        \]
        A matching with cost at most $\delta$ is called a \emph{$\delta$-matching}.
        
    \subparagraph*{Free space.}
        For $\delta \geq 0$, a point $(x, y) \in \D(\sigma, \tau)$ is \emph{$\delta$-free} if $d(\sigma(x), \tau(y)) \leq \delta$.
        The \emph{$\delta$-free space} $\F_\delta(\sigma, \tau)$ of $\sigma$ and $\tau$ is the subset of $\D(\sigma, \tau)$ consisting of all $\delta$-free points.
        We say that a point $(x, y)$ can \emph{$\delta$-reach} a point $(x', y')$ with $x \leq x'$ and $y \leq y'$ if and only if a $\delta$-matching between $\sigma[x, x']$ and $\tau[y, y']$ exists.
        Equivalently, we call $(x', y')$ \emph{$\delta$-reachable} from $(x, y)$.
    
    \subparagraph*{Fréchet distance.}
        The \f distance between~$\sigma$ and~$\tau$, denoted $\dF(\sigma,\tau)$, is the infimum~$\delta$ for which a $\delta$-matching exists.
        The \f distance for hybrid paths is symmetric and satisfies the triangle inequality, and assigns distance zero to paths that are orientation-preserving reparameterizations of one another.
        In fact, the \f distance is a metric for continuous hybrid paths up to orientation-preserving reparameterization.
        On the other hand, for general (not necessarily continuous) hybrid paths, it is unclear whether hybrid paths with zero \f distance are necessarily reparameterizations of one another.
        Still, the \f distance remains a pseudometric for general hybrid paths.
        Our definition of the \f distance for hybrid paths when restricted to polylines or discrete paths coincides with the respective existing definitions of the continuous and discrete \f distance.

\section{Input assumptions}
\label{sec:input_assumptions}
    In this work, we consider algorithms for approximating the \f distance between hybrid paths $\sigma$ and $\tau$ in some metric space $(M, d)$.
    Because the free space between hybrid paths in general metric spaces can be quite complex, we make some mild assumptions.

    \begin{figure}[t]
        \centering
        \includegraphics{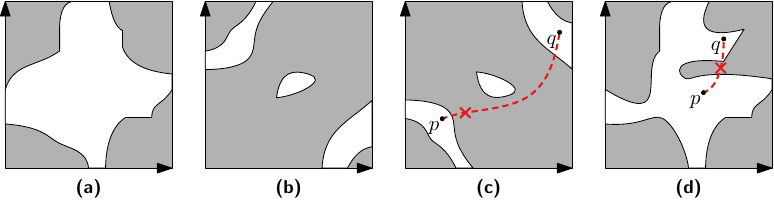}
        \caption{
        \lab{(a--b)} Two cells with the $\delta$-free-reachability property, with the $\delta$-free space drawn in white.
        \lab{(c--d)} Two cells that do not have the property, as witnessed by $p$ and $q$.}
        \label{fig:ortho-convex}
    \end{figure}

    A cell $C = X_i \times Y_j$ of the parameter space is said to have the \emph{$\delta$-free-reachability} property if, for all $\delta$-free points $(x, y)$ and $(x', y')$ in $C$, $(x, y)$ can $\delta$-reach $(x', y')$ precisely if $x \leq x'$ and $y \leq y'$.
    See \cref{fig:ortho-convex}.
    Such a property (or the stronger property of convexity of the $\delta$-free space in a cell) is used in most decision algorithms for computing the \f distance.

    \begin{remark}
        If a cell has the $\delta$-free-reachability property, the $\delta$-free space within is ortho-convex (though possibly disconnected).
        On the other hand, if the $\delta$-free space in a cell is connected and ortho-convex, the cell satisfies the $\delta$-free-reachability property.
    \end{remark}

    \begin{property}
    \label[property]{prop:ortho-convex}
        Let $\sigma$ and $\tau$ be hybrid paths.
        We say that $\F_\delta(\sigma, \tau)$ has the \emph{free-reachability} property if all cells of $\D(\sigma, \tau)$ have the $\delta$-free-reachability property.
    \end{property}
    Aside from the above property on the hybrid input paths and their free space, we require access to certain parts of the free space.
    To keep things general without requiring the entire free space as input, we access parts of it using the following oracle:
        
    \begin{definition}
    \label{def:free_space_oracle}
        A \emph{$\delta$-free space oracle} for hybrid paths $\sigma$ and $\tau$ supports the following queries.
        Given a cell $C$ of $\D(\sigma, \tau)$ and a horizontal or vertical line $\ell$, the oracle reports
        \begin{enumerate}
            \item the intersection of $\ell\cap C$ with the $\delta$-free space $\F_\delta(\sigma,\tau)$, as well as
            \item the intersection of $\ell\cap C$ with the boundary of the $\delta$-free space $\partial\F_\delta(\sigma,\tau)$.
        \end{enumerate}
    \end{definition}

    \noindent
    We assume access to such an oracle with query time $T_O = \Omega(1)$.
    
    \begin{remark}
        If $\D(\sigma, \tau)$ satisfies \cref{prop:ortho-convex}, then a $\delta$-free space oracle for query $(C,\ell)$ reports for \enumit{1.} at most a segment on $\ell$, and for \enumit{2.} at most two segments on $\ell$.
    \end{remark}

\section{Computing a surrogate}
\label{sec:surrogate}

In this section, we present an algorithm that, given hybrid paths $\sigma$ and $\tau$, roughly speaking can efficiently compute a sequence of hybrid subpaths of $\sigma$ that form a hybrid path within \f distance $\delta$ of $\tau$, if such a sequence exists.
In that case, we call the resulting sequence a surrogate of $\tau$ in $\sigma$.
Our primary purpose is to use surrogates for our $3$-approximate decision algorithm in \cref{sec:decision}.
However, surrogates may also be of independent interest.
For example, the surrogate may be of much lower complexity than~$\tau$, in which case the surrogate can be interpreted as a simplification of $\tau$, as exemplified in \cref{sec:decision_imbalanced}.

We now formally define the concept of a surrogate.
Let $\sigma$ be parameterized over $X \subseteq [0, 1]$.
For a sequence $X_1, \dots, X_k$ of interior-disjoint closed intervals, the restriction of $\sigma$ to the domain $X \cap (X_1\cup\dots\cup X_k)$, forms a \emph{$k$-piece hybrid subpath} of $\sigma$.
We define a \emph{$(\sigma, k, \delta)$-surrogate} of $\tau$ as a $k'$-piece hybrid subpath $\calS$ of $\sigma$, with $k' \leq k$, so that $\dF(\calS, \tau) \leq \delta$.

If there does not exist a hybrid subpath of $\sigma$ within \f distance $\delta$ of $\tau$, then it is unnecessary to compute a surrogate representing the entirety of $\tau$.
Instead, we require only a surrogate of some prefix $\tau[0, y']$ for which $\D(\sigma, \tau) \cap ([0, 1] \times [0, y'])$ contains all $\delta$-matchings that start on the bottom side of $\D(\sigma, \tau)$ and end on the top or right side.
We refer to prefixes $\tau[0, y']$ with this property as \emph{reachability-preserving}.
In this section, we present an algorithm that, given parameters $k$ and $\delta$, computes a $(\sigma, k, \delta)$-surrogate of an unspecified reachability-preserving prefix $\tau[0, y']$ of $\tau$.
In particular, a suitable value $y'$ is determined by the algorithm.
In addition, our algorithm computes a useful representation of a $\delta$-matching between the surrogate and~$\tau[0, y']$, which can be queried to determine the portion of $\tau$ matched to any given point on the surrogate.

We assume that $\sigma$ and $\tau$ respectively are hybrid $n$- and $m$-paths for which $\F_\delta(\sigma, \tau)$ satisfies \cref{prop:ortho-convex}.
Let $\pi^*$ be the leftmost $\delta$-matching that starts on the bottom side of~$\D(\sigma, \tau)$ and ends on the top or right side.
If $\pi^*$ ends in a point $(x^*, y^*)$, then $\tau[0, y^*]$ is the minimum reachability-preserving prefix of $\tau$.
(If $\pi^*$ does not exist, all prefixes of $\tau$ are reachability-preserving.)
Our algorithm computes a sequence of $\delta$-matchings $\pi_1, \dots, \pi_{k'}$, for some $k' \leq k$, so that every point on $\pi^*$ has a point on some $\pi_i$ horizontally left of it, and the first point on $\pi_i$ lies horizontally right of the last point on $\pi_{i-1}$.
These paths correspond to $k'$ hybrid subpaths of $\sigma$, which together form a $k'$-piece hybrid subpath~$\calS$.
Additionally, these paths correspond to a reachability-preserving prefix of $\tau$.
Specifically, if~$\pi_{k'}$ ends in a point $(x', y')$ (where $y' \geq y^*$ by construction), the paths correspond to the prefix~$\tau[0, y']$.
Thus, $\calS$ is a $(\sigma, k, \delta)$-surrogate of $\tau[0, y']$, and the paths correspond to a $\delta$-matching between~$\calS$~and~$\tau[0, y']$.

We compute paths $\pi_1, \dots, \pi_{k'}$ as above in the following manner.
Let $\tau$ be parameterized over $Y \subseteq [0, 1]$.
Let $[y_1, y'_1], \dots, [y_k, y'_k]$ be a sequence of intervals that covers $Y$, so that $y'_i \leq y_{i+1}$ for all $i$, and each hybrid subpath $\tau[y_i, y'_i]$ is a hybrid $m_i$-path for some $m_i = O(m/k)$.
We compute the paths $\pi_1, \dots, \pi_{k'}$ by iteratively applying the algorithm of \cref{lem:leftmost_matching} to the hybrid subpaths $\tau[y_i, y'_i]$.

\begin{lemma}
\label{lem:leftmost_matching}
    Let hybrid $n'$-path $\sigma'$ be a given suffix of $\sigma$.
    For a given~$j$, let $\pi$ (if one exists) be the leftmost $\delta$-matching that starts on the bottom side of $\D(\sigma', \tau[y_j, y'_j])$ and ends on either the top or right side.
    There is an algorithm with the following property:
    \begin{itemize}
        \item If $\pi$ does not exist, the algorithm reports this after $O(T_O \cdot n'm/k)$ time.
        \item If $\pi$ does exist, and ends in the $i$-th column of the parameter space grid, the algorithm instead takes $O(T_O \cdot im/k)$ time.
        It outputs a representation of $\pi$ that can be queried with a vertical line $\ell$ to compute $\ell \cap \pi$ in $O(T_O + \log nm)$ time.
    \end{itemize}
    The algorithm uses $O(n'm/k)$ space in the first case, and $O(im/k)$ in the second.
\end{lemma}
\begin{proof}
    We determine a sequence of intersection points between $\pi$ and the edges of the parameter space grid, one for each intersected edge.
    We use this sequence to represent $\pi$.
    To do so, slightly modify the standard dynamic programming algorithm by Alt and Godau.
    
    Their algorithm computes, for every edge of the grid, the subset of points that are $\delta$-reachable from points on the bottom side of the parameter space.
    It assumes that a point $p$ on the bottom or left side of a cell $C$ can $\delta$-reach a point $q$ on the top or right side of $C$ if and only if both $p$ and $q$ are $\delta$-free and $q$ lies above and to the right of $p$.
    This assumption is satisfied in our setting, since $\F_\delta(\sigma, \tau)$, and thus $\F_\delta(\sigma', \tau[y_j, y'_j])$, satisfies \cref{prop:ortho-convex}.
    We apply Alt and Godau's dynamic programming algorithm column-wise over the grid, and terminate early if an edge on the top side of the grid contains $\delta$-reachable~points.

    The number of cells considered by the algorithm is $O(im/k)$ (or $O(n'm/k)$ if $\pi$ does not exist).
    Each cell is processed in $O(1)$ time once the $\delta$-free points on the boundary of the cell are known.
    Using the free space oracle to compute these points, we process a cell in $O(T_O)$ time in total.
    This gives a total running time of $O(T_O \cdot im/k)$.
    We compute the intersection points between $\pi$ and the edges of the grid through backtracking, taking $O(im/k)$ additional time.
    The algorithm uses $O(1)$ space per cell, so the total space used is $O(im/k)$.

\subparagraph*{Intersections with lines.}
    We consider a query with a vertical line $\ell$, where we seek to compute the intersection $\ell \cap \pi$.
    Because $\pi$ is a $\delta$-matching, it is bimonotone and so the intersection is a vertical line segment.
    We show how to compute the bottommost intersection point; the topmost point, and thus the intersection, can be computed symmetrically.

    Let $\ell = \{x\} \times (-\infty, \infty)$ and let $p = (x, y)$ be the bottommost point in $\ell \cap \pi$.
    We first compute a cell $C$ of $\D(\sigma', \tau[y_j, y'_j])$ that contains $p$.
    If no cells contain $p$, we immediately return that the intersection is empty.
    We can compute $C$, or determine that it does not exist, in $O(\log (n'm/k)) = O(\log nm)$ time by performing binary search over the stored intersection points between $\pi$ and the edges of the grid.

    Let $\pi'$ be the subpath of $\pi$ inside $C$.
    This path is the top-leftmost $\delta$-matching between its endpoints.
    Hence, it is the composition of a vertical segment, a part of the boundary of $\F_\delta(\sigma', \tau[y_j, y'_j])$, and a horizontal segment.
    (Any of the three parts may be empty.)
    Let $(\hat{x}, \hat{y})$ and $(\hat{x}', \hat{y}')$ be the starting and ending points of $\pi'$.
    We distinguish three cases.
    
    If $\hat{x} = x$, we know $p = (x, \hat{y})$.
    If $(x, \hat{y}')$ is a point on the interior of $\F_\delta(\sigma', \tau[y_j, y'_j])$, then $p = (x, \hat{y}')$.
    Otherwise, the intersection between $\ell$ and $\pi'$ is a subset of the boundary of $\F_\delta(\sigma', \tau[y_j, y'_j])$.
    We can determine if this is the case in $O(T_O)$ time, by querying the $\delta$-free space oracle to determine if $(x, \hat{y}')$ is an interior point or not.
    Then, in the third case, $p$ is the bottom endpoint of the top component of $\ell \cap C \cap \partial \F_\delta(\sigma', \tau[y_j, y'_j])$.
    We compute this point with another query to the $\delta$-free space oracle.
\end{proof}
\begin{figure}[b]
    \vspace{-1em}
    \centering
    \includegraphics{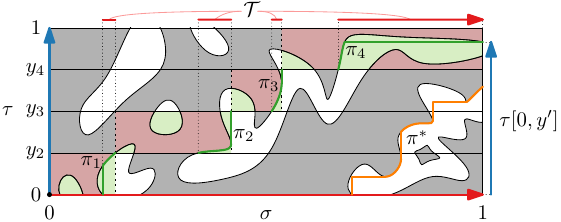}
    \caption{By exploring only the colored regions, the algorithm of \cref{thm:surrogate} (depicted for $k=4$) computes a $(k,\delta)$-surrogate~$\calT$ of a reachability-preserving prefix $\tau[0,y']$ of $\tau$ in $O(T_O \cdot nm/k)$ time.}
    \label{fig:surrogate}
\end{figure}
Iteratively, with $\ell$ ranging from $1$ to $k$, we apply the algorithm of \cref{lem:leftmost_matching} to $\tau[y_\ell, y'_\ell]$ and a shrinking suffix $\sigma[x,1]$, see \cref{fig:surrogate}.
We initialize $x \gets 0$, so that for $\ell=1$ we apply the algorithm to $\tau[y_1,y'_1]$ and the entirety of $\sigma$.
In iteration $\ell$, the algorithm computes (a representation of) the leftmost $\delta$-matching $\pi_\ell$ that starts on the bottom side of $\D(\sigma[x,1], \tau[y_\ell, y'_\ell])$ and ends on the top or right side of this parameter space (if such a matching exists).
If such a matching does not exist, we return (representations of) $\pi_1,\dots,\pi_{\ell-1}$.
Otherwise, we consider the end point of the $\delta$-matching $\pi_\ell$.
If $\pi_\ell$ ends at a point $(x,y)$ with $y<y'_\ell$, or $\ell=k$, then we return (representations of) $\pi_1,\dots,\pi_\ell$.
Otherwise, $\pi_\ell$ ends at a point $(x'_\ell,y'_\ell)$, and we update $x \gets x'_\ell$, and proceed with the next iteration, which applies the algorithm of \cref{lem:leftmost_matching} to $\tau[y_{\ell+1},y'_{\ell+1}]$ and the suffix $\sigma[x'_\ell,1]$.
As mentioned previously, the computed paths $\pi_1, \dots, \pi_{k'}$ correspond to a $\delta$-matching between a reachability-preserving suffix $\tau[0, y']$ and a $(\sigma, k, \delta)$-surrogate $\calS$ of it.
The total time spent computing the paths $O(T_O \cdot nm / k)$, and the algorithm of \cref{lem:leftmost_matching} uses $O(nm / k)$ space.

Because the projections of the paths $\pi_1, \dots, \pi_{k'}$ onto the horizontal axis are interior-disjoint, they correspond to hybrid subpaths of $\sigma$ that do not pairwise overlap, except possibly at their endpoints.
Thus, $\calS$ is a hybrid $n'$-path, for some $n' = O(n)$.
The paths $\pi_1, \dots, \pi_{k'}$ can be queried for intersections with vertical lines in $O(T_O + \log nm)$ time per query.
These representations can immediately be used to compute the hybrid subpath of $\tau$ matched to a given point on $\calS$.
We summarize:
    
\begin{theorem}
\label{thm:surrogate}
    Let $\tau$ be a subdivided path with $m$ breakpoints, such that $\F_\delta(\sigma, \tau)$ satisfies \cref{prop:ortho-convex}.
    Let $k \in [1, m]$ be a given integer parameter.
    In $O(T_O \cdot nm/k)$ time, and using $O(nm/k)$ space, we can compute the following:
    \begin{itemize}
        \item A $(\sigma, k, \delta)$-surrogate $\calS$ of a reachability-preserving prefix $\tau[0, y']$ of $\tau$.
        The surrogate is a hybrid $n'$-path for some $n' = O(n)$.
        \item A representation of a $\delta$-matching between $\calS$ and $\tau[0, y']$.
        This representation can be queried with a point on $\calS$ to obtain the hybrid subpath of $\tau$ matched to it.
        Queries take $O(T_O + \log nm)$~time.
    \end{itemize}
\end{theorem}

\section{An approximate decider}
\label{sec:decision}

We outline our strongly-subquadratic $3$-approximate decider for instances of the following~type.

\begin{definition}\label{def:decision_instance}
    Fix $n$ and $m\leq n$, and let $\sigma$ and $\tau$ respectively be hybrid $n$- and $m$-paths.
    Let $\delta \geq 0$ be a decision parameter for which $\F_\delta(\sigma, \tau)$, $\F_{2\delta}(\sigma, \sigma)$, and $\F_{2\delta}(\tau, \tau)$ all satisfy \cref{prop:ortho-convex}.
    We call the triple $(\sigma,\tau,\delta)$ a \emph{decision instance}.
\end{definition}
\begin{definition}\label{def:approx_decider}
    Given a decision instance $(\sigma,\tau,\delta)$, an \emph{$\alpha$-approximate decider} correctly reports either $\dF(\sigma, \tau) \leq \alpha\delta$ or $\dF(\sigma, \tau) > \delta$.
    It may report either answer if $\dF(\sigma, \tau) \in (\delta, \alpha\delta]$.
\end{definition}
For ease of exposition, we assume that $\sigma$ and $\tau$ are parameterized over domains $X$ and $Y$ that both contain the values $0$ and $1$.
Thus, the bottom-left corner of $\D(\sigma, \tau)$ is $(0, 0)$, and the top-right corner is $(1, 1)$.
Our algorithm works by determining whether $(0, 0)$ can $\delta$-reach $(1, 1)$, in an approximate manner.
Specifically, it either determines that $(1, 1)$ is $3\delta$-reachable, in which case we conclude that $\dF(\sigma, \tau) \leq 3\delta$, or it determines that $(1, 1)$ is not $\delta$-reachable, in which case $\dF(\sigma, \tau) > \delta$.
To do so, we use the concept of an \emph{approximate exit set}:

\begin{definition}[Entrance and exit set]
    An \emph{entrance set} $S$ for $\D(\sigma, \tau)$ is a set of points on the bottom and left sides of $\D(\sigma, \tau)$, with at most one connected component on each edge of the grid.
    For $\delta \geq 0$, the \emph{$\delta$-exit set} $\calR_\delta(S)$ of $S$ is the set of all points on the top and right sides of $\D(\sigma, \tau)$ that are $\delta$-reachable from points in $S$.
    Since $\F_\delta(\sigma, \tau)$ satisfies \cref{prop:ortho-convex}, $\calR_\delta(S)$ has complexity $O(n+m)$ and consists of at most one segment per edge of the grid.
    For $\alpha \geq 1$, an \emph{$(\alpha, \delta)$-exit set} of $S$ is an approximation of $\calR_\delta(S)$, defined as any set $A$ such that $\calR_\delta(S)\subseteq A\subseteq \calR_{\alpha\delta}(S)$, and $A$ consists of at most one segment on each edge of the grid.
\end{definition}
Our decision algorithm constructs a $(3, \delta)$-exit set for $\{(0, 0)\}$, and reports whether $(1, 1)$ is in this set.
We do so using a data structure for constructing (exact) exit sets with respect to a preprocessed hybrid path and a given hybrid subpath.
We present this data structure in \cref{sub:reachability_DS}.
Then, in \cref{sub:arbitrary_queries}, we extend this data structure to handle arbitrary hybrid query paths, though in an approximate manner.
Finally, in \cref{sub:approx_decider}, we preprocess various hybrid subpaths of $\sigma$ and~$\tau$ into these data structures, and use them to construct a $(3, \delta)$-exit set for $\{(0, 0)\}$.
See \cref{fig:algo} for an illustration of the algorithm.

\begin{figure}[t]
    \centering
    \includegraphics{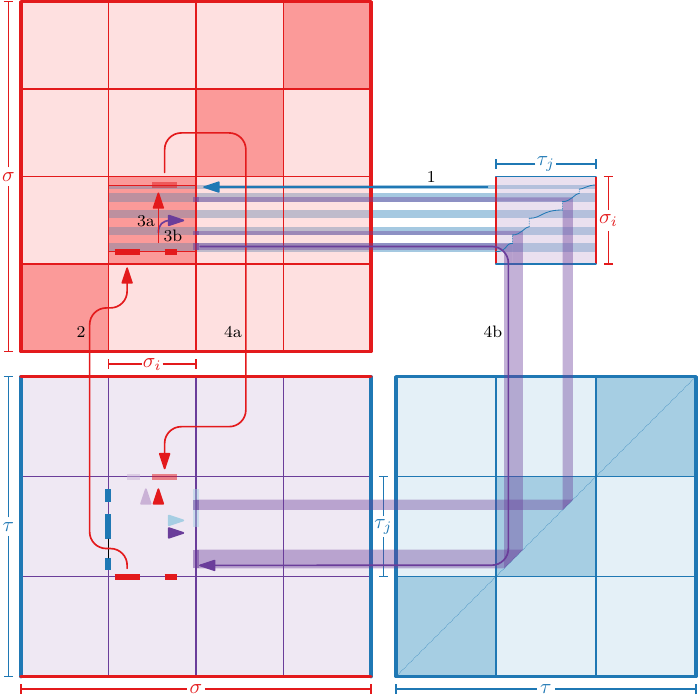}
    \caption{For each $i$, respectively $j$, we preprocess the free space $\F_{2\delta}(\sigma_i,\sigma_i)$ (dark red), respectively $\F_{2\delta}(\tau_j,\tau_j)$ (dark blue).
    For each $i$ and $j$, we use the corresponding information to 3-approximately propagate reachability across~$\F_{\delta}(\sigma_i,\tau_j)$ in two parts: from the bottom to the top and right (illustrated), and from the left to the top and right.
    To propagate an entrance set on the bottom~(2), we first (1) compute a sequence $\calS$ of subcurves of $\sigma_i$ that acts as a surrogate for $\tau_j$.
    Using the preprocessed information on $\F_{2\delta}(\sigma_i,\sigma_i)$, we propagate $2\delta$-reachability across $\sigma_i\times\calS$ to obtain a $2\delta$-exit set on the top (3a) and right (3b) of $\sigma_i\times\calS$.
    We can translate these exit sets back into $(3,\delta)$-exit sets on the top (4a) and right (4b) of $\sigma_i\times\tau_j$.}
    \label{fig:algo}
\end{figure}

\subsection{A data structure for exact exit sets}
\label{sub:reachability_DS}

In this section we develop our data structure for constructing exact exit sets with respect to a preprocessed hybrid path $\sigma$ and a given hybrid subpath of $\sigma$.
Let $\sigma$ be a hybrid $n$-path and let $\delta \geq 0$ be a parameter for which $\F_\delta(\sigma, \sigma)$ satisfies \cref{prop:ortho-convex}.
We present a data structure built on $\sigma$ and $\delta$ that, given a hybrid query subpath $\sigma'$ of $\sigma$ and an entrance set $S$ for $\D(\sigma, \sigma')$, efficiently reports the (exact) $\delta$-exit set $\calR_{2\delta}(S)$.
This data structure is used in steps (3a) and (3b) of \cref{fig:algo}.

We make use of a structure we call a \emph{reachability map}.
Let $\sigma'$ be a hybrid subpath of $\sigma$ with $m \leq n$ pieces, so $\sigma'$ is a hybrid $m$-path.
Let $\BL$ and $\TR$ be the unions of the bottom and left, respectively top and right, sides of $\D(\sigma, \sigma')$.
Let $\calR$ be the $\delta$-exit set of the set of all $\delta$-free points on $\BL$.
The \emph{$\delta$-reachability map} $\RM_\delta(\sigma, \sigma')$ of $\D(\sigma, \sigma')$ consists of a pair of functions $\alpha, \beta \from \BL \to \TR$, together with $\calR$, such that a point $p \in \BL$ can $\delta$-reach a point $q \in \TR$ precisely if $q \in \calR$ and $q$ lies between $\alpha(p)$ and $\beta(p)$ on~$\TR$.

Alt and Godau~\cite{alt95continuous_frechet} present a reachability data structure that is a $\delta$-reachability map, only with a more concrete representation.
While they consider only the free space of polylines in real-valued spaces, which is convex in every cell of the parameter space, their techniques trivially extend to our setting.
This is because their reachability data structure uses only the fact that a $\delta$-free point $p = (x, y)$ can $\delta$-reach a $\delta$-free point $q = (x', y')$ in the same cell precisely when $x \leq x'$ and $y \leq y'$.
This is exactly the free-reachability property (\cref{prop:ortho-convex}) that we assumed for $\F_\delta(\sigma, \sigma)$ and thus $\F_\delta(\sigma, \sigma')$.

Alt and Godau's results imply that the functions $\alpha$ and $\beta$ of $\RM_\delta(\sigma, \sigma')$ are piecewise-linear with $O(nm)$ pieces.
Thus, we can evaluate these functions for any given point on $\BL$ in $O(\log nm) = O(\log n)$ time through e.g.~binary search over the pieces.
Additionally, Alt and Godau's construction implies that we can construct $\RM_\delta(\sigma, \sigma')$ in $O(nm \log n)$ time, given the intersection between $\F_\delta(\sigma, \sigma')$ and the parameter space grid.
We can compute this intersection with $O(nm)$ queries to the $\delta$-free space oracle.
Alternatively, we can construct $\RM_\delta(\sigma, \sigma')$ faster by merging existing reachability maps, which circumvents the need to compute the free space at the grid edges.

\begin{lemma}
\label{lem:reachability_map_construction}
    The $\delta$-reachability map $\RM_\delta(\sigma, \sigma')$ has $O(nm)$ complexity.
    We can construct $\RM_\delta(\sigma, \sigma')$ in $O(nm \cdot (T_O + \log n))$ time, or in $O(nm)$ time when $\RM_\delta(\sigma, \sigma'[0, x])$ and $\RM_\delta(\sigma, \sigma'[x, 1])$ are given for some $x \in [0, 1]$.
\end{lemma}

\begin{lemma}
\label{lem:exit_sets}
    Given $\RM_\delta(\sigma, \sigma')$, we can compute the $\delta$-exit set of any entrance set in $O(n \log n)$ time.
\end{lemma}
\begin{proof}
    Let $S$ be an entrance set of $\D(\sigma, \sigma')$.
    For a point $p \in \BL$, let $\calR_p$ be the region between the points $\alpha(p)$ and $\beta(p)$.
    By definition of $\RM_\delta(\sigma, \sigma')$, we have $\calR_\delta(\{p\}) = \calR \cap \calR_p$ for every point $p$ on $\BL$.
    We compute $\calR_\delta(S) = \bigcup_{p \in S} \calR_\delta(\{p\})$ by first computing the union of the sets $\calR_p$ for the points $p \in S$.
    This union has $O(|S|) = O(n)$ components and can be computed in $O(n \log n)$ time through $\RM_\delta(\sigma, \sigma')$.
    Then, we compute the intersection of this union with $\calR$ to obtain $\calR_\delta(S)$, taking $O(n \log n)$ additional time, by sorting and scanning over the components in the sets.
\end{proof}
We wish to compute $\delta$-exit sets with respect to $\D(\sigma, \sigma')$, without specifying $\sigma'$ during preprocessing.
We can answer such queries efficiently using multiple reachability maps, as we describe next.
Consider a balanced binary tree $\calT$ built on the pieces of $\sigma$.
Each node $\mu$ of $\calT$ has a hybrid subpath $\sigma_\mu$ of $\sigma$ associated with it, namely the hybrid subpath made up of the pieces stored in the subtree rooted at $\mu$.
We store the $\delta$-reachability map $\RM_\delta(\sigma, \sigma_\mu)$ in node $\mu$.
The augmented tree $\calT$ is the data structure we use to compute exit sets.

\begin{lemma}
\label{lem:exit_set_DS}
    In $O(n^2 \cdot (T_O + \log n))$ time, we can build a data structure of $O(n^2 \log n)$ size that, given an interval $[x, x'] \subseteq [0, 1]$ and an entrance set $S$ of $\D(\sigma, \sigma[x, x'])$, reports the $\delta$-exit set of $S$ with respect to $\sigma$ and $\sigma[x, x']$ in $O(n \cdot (T_O + \log^2 n))$ time.
\end{lemma}
\begin{proof}
    We first argue the complexity and construction time of the augmented tree $\calT$.
    Every leaf node $\mu$ of $\calT$ is associated with a piece $\sigma_\mu$ of $\sigma$, which are all hybrid $1$-paths.
    Thus, each $\delta$-reachability map $\RM_\delta(\sigma, \sigma_\mu)$ has $O(n)$ complexity and takes $O(n \cdot (T_O + \log n))$ time to construct, by \cref{lem:reachability_map_construction}.
    This sums up to $O(n^2)$ size and $O(n^2 \cdot (T_O + \log n))$ construction time for all leaves.
    For an interior node $\mu$ with $m$ nodes in its subtree, the hybrid subpath $\sigma_\mu$ has $O(m)$ pieces.
    By \cref{lem:reachability_map_construction}, we can therefore compute $\RM_\delta(\sigma, \sigma_\mu)$ in $O(nm)$ time by merging the reachability maps stored in the child nodes of $\mu$.
    The size of the resulting reachability map is also $O(nm)$.
    The total time spent merging reachability maps is $O(n^2 \log n)$, implying the total construction time of $\calT$ is $O(n^2 \cdot (T_O + \log n))$.
    The total size of the reachability maps is $O(n^2 \log n)$.

    To answer a query for a given interval $[x, x'] \subseteq [0, 1]$ and entrance set~$S$ of $\D(\sigma, \sigma[y, y'])$,
    we seek to compute the $\delta$-exit set of $S$ with respect to $\sigma$ and $\sigma[x, x']$.
    Let $x < x_i < \dots < x_{i'} < x'$ index the breakpoints of $\sigma[x, x']$.
    Note that $x_i, \dots, x_{i'}$ are also breakpoints of $\sigma$.
    Thus, $\sigma[x_i, x_{i'}]$ is the concatenation of $O(\log n)$ hybrid subpaths $\sigma_\mu$ stored in $\calT$.
    
    We compute the $\delta$-exit set of $S$ by first computing its $\delta$-exit set with respect to $\sigma$ and $\sigma[x, x_i]$ (or $\sigma[x, x']$ if $x_i$ does not exist).
    By \cref{lem:reachability_map_construction,lem:exit_sets}, this takes $O(n \cdot (T_O + \log n))$ time by constructing $\RM_\delta(\sigma, \sigma[x, x'])$ and then constructing the exit set.
    This set induces an entrance set $S'$ for $\D(\sigma, \sigma[x_i, x_{i'}])$.
    Then, we compute the $\delta$-exit set of $S'$ with respect to $\D(\sigma, \sigma[x, x'])$ by querying $O(\log n)$ reachability maps in $\calT$.
    This takes $O(n \log^2 n)$ time in total.
    Finally, we turn this exit set into an entrance set $S''$ for $\D(\sigma, \sigma[x_{i'}, x'])$ and proceed as in the first step.
    The result is the $\delta$-exit set of $S$ (with respect to $\sigma$ and $\sigma'$).
\end{proof}

\subsection{Handling arbitrary hybrid query paths}
\label{sub:arbitrary_queries}

Next we extend our data structure for exact exit sets with respect to a preprocessed hybrid path $\sigma$ and hybrid subpath of $\sigma$, to construct approximate exit sets with respect to $\sigma$ and an arbitrary hybrid query path $\tau$.
These queries are exactly steps (3a) and (3b) of \cref{fig:algo}.
Note that we restrict the types of entrance sets we support.
Specifically, given $\sigma$ and $\delta$, we construct a data structure that, given $\tau$ and some entrance set $S$ on the bottom side of $\D(\sigma, \tau)$, computes a $(3, \delta)$-exit set for $S$.
Throughout this section, let $\sigma$ be a hybrid $n$-path and let $\delta \geq 0$ be a parameter for which $\F_{2\delta}(\sigma, \sigma)$ satisfies \cref{prop:ortho-convex}.
Our data structure will support querying with any hybrid path $\tau$ for which $\F_\delta(\sigma, \tau)$ satisfies \cref{prop:ortho-convex}.

For preprocessing, we construct the data structure of \cref{lem:exit_set_DS} on $\sigma$ and the value $2\delta$.
This data structure allows us to construct exact $2\delta$-exit sets of any entrance set, with respect to $\sigma$ and any hybrid query subpath of $\sigma$.
Next we present the query algorithm.

Let $\tau$ be the given hybrid query $m$-path and let $S$ be the given entrance set on the bottom side of $\D(\sigma, \tau)$.
Observe that to compute a $(3, \delta)$-exit of $S$, it is sufficient to consider any reachability-preserving prefix of $\tau$, instead of all of $\tau$.
That is, any prefix $\tau[0, y']$ of $\tau$ for which $\D(\sigma, \tau) \cap ([0, 1] \times [0, y'])$ contains all $\delta$-matchings that start on the bottom side of $\D(\sigma, \tau)$ and end on the top or right side.
This is because the exact $\delta$-exit set of $S$ is contained in the region $[0, 1] \times [0, y']$.
Our approach is therefore to compute a surrogate of some reachability-preserving prefix of $\tau$, in terms of hybrid subpaths of $\sigma$.
Then, we can use our data structure for exact exit sets from \cref{lem:exit_set_DS}.

Choose an integer parameter $k \in [1, m]$.
We compute a $(\sigma, k, \delta)$-surrogate of some reachability-preserving prefix $\tau[0, y']$ of $\tau$, using the algorithm of \cref{thm:surrogate}.
This takes $O(T_O \cdot nm / k)$ time and uses $O(nm / k)$ space.
The algorithm computes such a surrogate $\calS$, as well as some representation of a $\delta$-matching between $\calS$ and $\tau[0, y']$.
Importantly, this representation can be queried with a point on $\calS$ to obtain the hybrid subpath of $\tau$ matched to it.
Such queries take $O(T_O + \log nm)$ time.

We consider the parameter space $\D(\sigma, \calS)$ of $\sigma$ and $\calS$.
The entrance set $S$ of $\D(\sigma, \tau)$, which lies on the bottom side of the parameter space, corresponds to an entrance set $S'$ of $\D(\sigma, \calS)$, also on the bottom side.
We compute the $2\delta$-exit set $\calR_{2\delta}(S')$ of $S'$ with respect to $\D(\sigma, \calS)$.
We do so by querying the data structure of \cref{lem:exit_set_DS} a total of at most $k$ times, since the surrogate $\calS$ is formed by a sequence of at most $k$ hybrid subpaths of $\sigma$.
These queries take a total of $O(kn \cdot (T_O + \log^2 n))$ time.

Each point $(x, y) \in \calR_{2\delta}(S')$ corresponds to a vertical line segment $\{x\} \times \mu(y) \in \D(\sigma, \tau)$, where $\mu(y) = [\hat{y}, \hat{y}']$ is the interval indexing the hybrid subpath of $\tau$ matched to $\calS(y)$.
That is, $\tau[\hat{y}, \hat{y}']$ is matched to $\calS(y)$ in the $\delta$-matching between $\calS$ and $\tau$.
We define $A = \bigcup_{(x, y) \in \calR_{2\delta}(S')} \mleft( \{x\} \times \mu(y) \mright)$ and show it is a $(3, \delta)$-exit set for $S$ with respect to $\sigma$ and~$\tau$:

\begin{lemma}
    The set $A$ as defined above is a $(3, \delta)$-exit set for $S$.
\end{lemma}
\begin{proof}
    We assume without loss of generality that $\calS$ has been re-parameterized over $[0, 1]$ so that $\mu(y) = \{y\}$ for all $y \in [0, 1]$.
    That is, the point $\calS(y)$ is $\delta$-matched to the point $\tau(y)$, and only that point.
    With this new parameterization, we have
    \(
        \mleft| d(\sigma(x), \calS(y)) - d(\sigma(x), \tau(y)) \mright| \leq \delta
    \)
    for all $x, y \in [0, 1]$.
    Thus, if a point $(x, y) \in \D(\sigma, \tau)$ can $\delta$-reach a point $(x', y') \in \D(\sigma, \tau)$, then $(x, y) \in \D(\sigma, \calS)$ can $2\delta$-reach $(x', y') \in \D(\sigma, \calS)$.
    Conversely, if a point $(x, y) \in \D(\sigma, \calS)$ can $2\delta$-reach a point $(x', y') \in \D(\sigma, \calS)$, then $(x, y) \in \D(\sigma, \tau)$ can $3\delta$-reach $(x', y') \in \D(\sigma, \tau)$.
    It follows that $A$ contains all points that are $\delta$-reachable from points in $S$, and only points that are $3\delta$-reachable, making it a $(3, \delta)$-exit set.
\end{proof}
It remains to compute $A$.
Given $\calR_{2\delta}(S')$, which consists of $O(n + k) = O(n + m)$ horizontal and vertical line segments, we compute $\mu(p)$ for each endpoint $p$ of these line segments.
We do so by querying the representation of the $\delta$-matching between $\calS$ and $\tau$, taking $O(T_O + \log nm)$ time each.
The image of $\mu$ over a line segment is the convex hull of the image over the endpoints, and so we compute $A$ in $O((n+m) \cdot (T_O + \log nm))$ time.

In summary, we use $O(T_O \cdot nm / k)$ time (and $O(nm / k)$ space) to compute a surrogate of~$\tau$; $O(kn \cdot (T_O + \log^2 n))$ time to compute $\calR_{2\delta}(S')$; and $O((n+m) \cdot (T_O + \log nm))$ time to compute $A$.
To optimize the running time, we set $k=\sqrt{\frac{T_O \cdot m}{T_O + \log^2 n}}=\tilde \Theta(\sqrt m)$.
\begin{lemma}
\label{lem:approx_exit_set_DS}
    Suppose $\sigma$ is preprocessed into the data structure of \cref{lem:exit_set_DS}, with the parameter $2\delta$.
    Let $\tau$ be a hybrid $m$-path for which $\F_\delta(\sigma, \tau)$ satisfies \cref{prop:ortho-convex}.
    Let $k \in [1, m]$ be an integer parameter.
    Let $S$ be an entrance set on the bottom side of $\D(\sigma, \tau)$.
    We can compute a $(3, \delta)$-exit set for $S$ in $O(T_O \cdot nm / k + kn \cdot (T_O + \log^2 n) + (n+m) \cdot (T_O + \log nm))$ time, using $O(nm / k)$ space.
\end{lemma}

\begin{corollary}
\label{cor:approx_exit_set_DS}
    In the same context as \cref{lem:approx_exit_set_DS}, setting $k = \tilde \Theta(\sqrt m)$,
    we can compute a $(3, \delta)$-exit set for $S$ in $\tilde O(T_O \cdot n\sqrt{m})$ time and $\tilde O(n\sqrt{m})$ space.
\end{corollary}

\subsection{Putting things together}
\label{sub:approx_decider}
    Constructing our data structure for approximate exit sets from \cref{sub:arbitrary_queries} takes super-quadratic time.
    Thus, the data structure by itself does not result in a decision algorithm that runs in strongly-subquadratic time.
    In \cref{lem:approx_exit_sets}, we construct the data structure for several small hybrid subpaths, ensuring that the total preprocessing time is strongly-subquadratic,
    and use these data structures to construct a $(3, \delta)$-exit set for any given entrance set of~$\D(\sigma, \tau)$.
    
    \begin{lemma}
    \label{lem:approx_exit_sets}
        Let $k_\sigma \in [1, n]$ and $k_\tau \in [1, m]$ be integer parameters.
        Given an entrance set~$S$ of $\D(\sigma, \tau)$, we can compute a $(3, \delta)$-exit set for $S$ in
        \begin{align*}
            O{\mleft(
                (n^2 / k_\sigma + m^2 / k_\tau) (T_O + \log n) + 
                (n \sqrt{m k_\tau} + m \sqrt{n k_\sigma}) (T_O+\sqrt{T_O} \log n)
            \mright)}
        \end{align*}
        time and $O((n^2 / k_\sigma + m^2 / k_\tau) \log n)$ space.
    \end{lemma}
    \begin{proof}
        Let $\hat{n} = n / k_\sigma$ and $\hat{m} = m / k_\tau$.
        We partition $\sigma$ (respectively $\tau$) into hybrid subpaths $\sigma_1, \dots, \sigma_{k_\sigma}$ (respectively $\tau_1, \dots, \tau_{k_\tau}$) with $O(\hat{n})$ (respectively $O(\hat{m})$) pieces each.
        We preprocess each of these hybrid subpaths into the data structure of \cref{lem:exit_set_DS}, taking $O(\hat{n}^2 \cdot (T_O + \log \hat{n}))$ time for each $\sigma_i$, and $O(\hat{m}^2 \cdot (T_O + \log \hat{m}))$ time for each $\tau_j$.
        The space used is $O(\hat{n}^2 \log \hat{n})$ for each $\sigma_i$, and $O(\hat{m}^2 \log \hat{m})$ for each $\tau_j$.
        Thus, the total preprocessing time is $O((n^2 / k_\sigma + m^2 / k_\tau) \cdot (T_O + \log n))$, and the space used is $O((n^2 / k_\sigma + m^2 / k_\tau) \log n)$.
        
        With the data structures built, we compute a $(3, \delta)$-exit set for the given entrance set~$S$.
        The partitions of $\sigma$ and $\tau$ induce a partition of $\D(\sigma, \tau)$ into a grid with cells $\calR_{i, j}$ that correspond to $\sigma_i$ and $\tau_j$.
        See \cref{fig:algo} (bottom left).
        We process these regions column-wise, from bottom to top.
        For each region $\calR_{i, j}$, we compute a $(3, \delta)$-exit set $A_{i, j}$ for $S$ with respect to the prefixes of $\sigma$ and $\tau$ up to and including $\sigma_i$ and $\tau_j$.
        Given such exit sets $A_{i-1, j}$ and $A_{i, j-1}$ (where $A_{0, j} = \emptyset$ and $A_{i, 0} = \emptyset$), we compute $A_{i, j}$ as follows.

        The set $S \cup A_{i-1, j}$ induces an entrance set $S_\ell$ on the left side of $\calR_{i, j}$, and $S \cup A_{i, j-1}$ induces an entrance set $S_b$ on the bottom side.
        Thus, we can query the data structures built on $\sigma_i$ and~$\tau_j$ with $S_b$ and $S_\ell$, respectively, to obtain $(3, \delta)$-exit sets for $S_b$ and $S_\ell$, whose union is~$A_{i, j}$.
        The queries allow specifying a parameter $k$ in $[1, n]$ or $[1, m]$, respectively, to influence the running time and space usage of the query algorithm.
        For querying the data structure built on $\sigma_i$, we set $k = \sqrt{\frac{T_O \cdot m}{T_O + \log^2 n}}$, and for querying the data structure built on~$\tau_j$, we set $k = \sqrt{\frac{T_O \cdot n}{T_O + \log^2 m}}$.
        This results in the queries running in
        \begin{align*}
            &O{\mleft( (\hat{n}+\hat{m}) (T_O + \log \hat{n}\hat{m}) + \hat{n} \sqrt{\hat{m}} \ (T_O+\sqrt{T_O} \log \hat{n}) \mright)}
            \text{ and}\\
            &O{\mleft( (\hat{n}+\hat{m}) (T_O + \log \hat{n}\hat{m}) + \hat{m} \sqrt{\hat{n}} \ (T_O+\sqrt{T_O} \log \hat{m}) \mright)}
        \end{align*}
        time, respectively, using $O(\hat{n} \sqrt{\hat{m}} + \hat{m} \sqrt{\hat{n}})$ space.
        Note that the space used by each individual query is dominated by the total space used by the space used by the larger of the queried data structures, and thus can be disregarded.
        We obtain $A_{i,j}$ by taking the union of the computed $(3, \delta)$-exit sets in $O(\hat{n}+\hat{m})$ extra time.
        
        After preprocessing the hybrid subpaths $\sigma_i$ and $\tau_j$ into the data structures, our algorithm makes $k_\sigma \cdot k_\tau$ queries to these data structures.
        These queries in total take time
        \begin{align*}
            O( (nk_\tau + mk_\sigma) (T_O + \log nm) + (n\sqrt{m k_\tau}+ m \sqrt{n k_\sigma}) (T_O+\sqrt{T_O} \log n) ).
        \end{align*}
        Since $m \leq n$, we have $\log nm=O(\log n)$.
        Also, since $k_\sigma \leq \sqrt{nk_\sigma}$ and $k_\tau \leq \sqrt{mk_\tau}$, we have $(nk_\tau + mk_\sigma) (T_O + \log nm) = O((n\sqrt{m k_\tau}+ m \sqrt{n k_\sigma}) (T_O+\sqrt{T_O} \log n))$.
        Including the time spent constructing the data structures, we arrive at the claimed running time.
    \end{proof}
    We subsequently obtain a $3$-approximate decider by testing whether the top-right corner lies in the exit set.
    Setting $k_\sigma \approx \frac{n}{m^{2/3}} \cdot \frac{(T_O + \log n)^{2/3}}{(T_O^2 + T_O \log^2 n)^{1/3}}$ and $k_\tau \approx \frac{m}{n^{2/3}} \cdot \frac{(T_O + \log n)^{2/3}}{(T_O^2 + T_O \log^2 n)^{1/3}}$ gives:
    \begin{theorem}\label{thm:approx_decider}
        There is a $3$-approximate decider (as per \cref{def:approx_decider}) running in time ${O(nm^{2/3} (T_O^3 + T_O^2 \log^2 n + T_O \log^3 n)^{1/3})}$ and space $O(nm^{2/3} (T_O^3 + T_O^2 \log^2 n + T_O \log^3 n)^{1/3})$.
    \end{theorem}

\subsection{Space-efficiency}\label{sec:time_and_space}
    Although we aimed to optimize time complexity in \cref{lem:approx_exit_sets}, the corresponding space complexity is likely impractical.
    The space usage of our approximate decider is dominated by the 
    precomputed data structures for $\F_{2\delta}(\sigma_i,\sigma_i)$ for $1\leq i\leq k_\sigma$ and $\F_{2\delta}(\tau_j,\tau_j)$ for $1\leq j\leq k_\tau$.
    In total, the data structures for~$\sigma$ use $\tilde\Theta(n^2/k_\sigma)$ space, and those for $\tau$ use $\tilde\Theta(m^2/k_\tau)$ space.

    Our algorithm propagates reachability across  $\sigma_i\times\tau_j$ for each pair $(i,j)$.
    For the correctness of the algorithm, the order in which these pairs are processed is largely irrelevant.
    The only constraint is that the pair $(i,j)$ should not be processed before $(i-1,j)$ and $(i,j-1)$.
    
    In particular, we can propagate reachability using two nested for-loops, the outer loop iterating over~$j$, and the inner loop over~$i$.
    The data structure for $\F_{2\delta}(\tau_j,\tau_j)$ is then only necessary during a single iteration of the outer loop, and we can compute it at the start of that iteration, and discard it at the end of the iteration.
    Doing so, the total running time remains $\tilde O(T_O \cdot (n^2/k_\sigma + m^2/k_\tau + n \sqrt{m k_\tau} + m\sqrt{n k_\sigma} ))$, but the space reduces to $\tilde O(n^2/k_\sigma + (m/k_\tau)^2)$.
    Setting $k_\sigma=\Theta(n/m^{2\eps})$ and $k_\tau=\Theta(m^{1-2\eps})$ for $\eps\in[0,\frac{1}{3}]$ yields a space-time tradeoff:
    \begin{theorem}\label{thm:approx_decider_tau_DS_during_iteration}
        For any $\eps\in[0,\frac{1}{3}]$, there is a $3$-approximate decider using $\tilde O(T_O \cdot nm^{1-\eps})$ time and $\tilde O(nm^{2\eps})$~space.
    \end{theorem}
    We can moreover recompute the data structure for~$\sigma_i$ in each iteration of the inner loop.
    Doing so uses $\tilde O(T_O \cdot (\frac{m^2}{k_\tau} + \frac{n}{k_\sigma} \frac{m}{k_\tau} (\frac{n}{k_\sigma})^2 + n \sqrt{m k_\tau} + m\sqrt{n k_\sigma} ))$ time and $\tilde O(n + (\frac{n}{k_\sigma})^2 + (\frac{m}{k_\tau})^2)$ space.
    For any $\eps\in[0,\frac{2}{9}]$, selecting $k_\sigma=\Theta(n/m^{2\eps})$ and $k_\tau=\Theta(m^{1-2\eps})$ gives $\tilde O(T_O \cdot nm^{1-\eps})$ time and $\tilde O(n)$ space.
    Plugging in $\eps=\frac{2}{9}$ gives $k_\sigma=\Theta(n/m^{\frac{4}{9}})$ and $k_\tau=\Theta(m^{\frac{5}{9}})$, and yields:
    \begin{theorem}\label{thm:approx_decider_both_DS_during_iteration}
        There is a $3$-approximate decider using $\tilde O(T_O \cdot nm^{7/9})$ time and $\tilde O(n)$ space.
    \end{theorem}
    
\section{Approximating the Fréchet distance}
\label{sec:optimization}

We now discuss how to approximate the \f distance between two hybrid paths $\sigma$ and~$\tau$ in a metric space $(M, d)$.
If $\sigma$ and $\tau$ are polylines in Euclidean space, there exists a black-box algorithm that takes any approximate decider and computes an approximation of the \f distance, at only a slight increase in approximation factor and running time~\cite{colombe21continuous_frechet}.
A more general procedure that does not rely on a specific metric space is not yet known however.
We present an alternative approach for hybrid paths whose pieces are continuous that, aside from \cref{prop:ortho-convex}, does not require any structure on the underlying metric space.

In \cref{sub:prefix-length_matching}, we introduce a useful class of matchings and analyze their cost.
We relate this cost to the \f distance, and in \cref{sub:optimization} use it to compute a suitable over-estimate of the \f distance.
This over-estimate is used to initiate a search procedure.
The search procedure uses any $\alpha$-approximate decider to compute an $(\alpha+\eps)$-approximation of $\dF(\sigma, \tau)$, for any $\eps > 0$.
In \cref{sub:specific_spaces}, we analyze the complexity of our algorithm in specific settings.

\subsection{Matchings near the diagonal}
\label{sub:prefix-length_matching}

Define the length of a hybrid path $\sigma \from X \to M$ in a metric space $(M, d)$ as
\[
    \|\sigma\| := \sup \ \sum_{i=1}^k d(\sigma(x_i), \sigma(x_{i+1})),
\]
where the supremum ranges over all finite increasing sequences $x_1 < \dots < x_k$ in $X$.
Intuitively, we measure the lengths of all pieces, including discontinuous ``jumps'' in and between pieces.
Define $\ell_\sigma(x)=\|\sigma[0,x]\|$ and $\ell_\sigma([x,x'])=[\ell_\sigma(x),\ell_\sigma(x')]$.

We show how to get a suitable over-estimate of $\dF(\sigma, \tau)$ by matching points based on how far they are along the hybrid paths.
Ideally, we would match points of $\sigma$ and $\tau$ at the same arclength, but this is not always possible due to discontinuities and paths having distinct lengths.
As a compromise, we consider matchings in which prefixes of matched points differ in length by at most $\beta$.
Specifically, we say that a matching between $\sigma$ and $\tau$ is \emph{$\beta$-diagonal} if $|\ell_\sigma(x) - \ell_\tau(y)| \leq \beta$ for all matched $x$ and $y$.
See \cref{fig:length_matching}.
In \cref{lem:comparing_diagonal_matchings}, we use the concept of $\beta$-diagonal matchings to relate the costs of different matchings.

\begin{lemma}
\label{lem:comparing_diagonal_matchings}
    Let $\sigma$ and $\tau$ be hybrid paths.
    Let $(f,g)$ and $(f',g')$ be $\beta$- and $\beta'$-diagonal matchings of cost $\delta$ and $\delta'$, respectively.
    Then $|\delta-\delta'|\leq \beta+\beta'$.
\end{lemma}
\begin{proof}
    Consider an arbitrary $(x,y')$ matched by $(f',g')$.
    If $(x,y)$ are matched by $(f,g)$, then
    \begin{align*}
        d(\sigma(x),\tau(y'))
        \leq d(\sigma(x),\tau(y)) + d(\tau(y),\tau(y')) 
        &\leq \delta + \|\tau[y,y']\| \\
        &=    \delta + |\ell_\tau(y) - \ell_\tau(y')| \\
        &\leq \delta + (\ell_\sigma(x)+\beta)-(\ell_\sigma(x)-\beta') \\
        &=    \delta + \beta + \beta'.
    \end{align*}
    Because $(x,y')$ was an arbitrary pair matched by $(f',g')$, it follows that $\delta'\leq\delta+\beta+\beta'$.
    By a symmetric argument, we also obtain $\delta\leq\delta'+\beta+\beta'$, and the claim follows.
\end{proof}
\begin{figure}[h]
    \centering
    \includegraphics{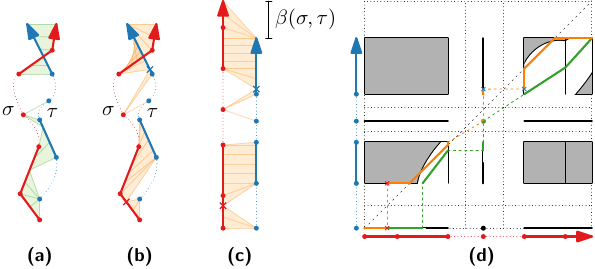}
    \caption{
        \lab{(a)} An optimal $\delta$-matching between two hybrid paths $\sigma$ and $\tau$.
        \lab{(b--c)} A $\beta(\sigma,\tau)$-diagonal matching that stays as close to the diagonal as possible.
        \lab{(d)} The $\delta$-free space, drawn parameterized by length, together with the matchings from \lab{(a)} (green) and \lab{(b--c)} (orange).}
    \label{fig:length_matching}
\end{figure}
A piece of a hybrid path can be arbitrarily long compared to its diameter.
Consequently, the cost of a $\beta$-diagonal matching can be arbitrarily much larger than the \f distance.
To bound the cost, we use an additional parameter $r \geq 0$ that captures how much the pieces of $\sigma$ and $\tau$ resemble shortest paths.
We say that a piece (of $\sigma$ or $\tau$) with endpoints $p$ and $q$ is \emph{$r$-shortest} if its length is at most $d(p, q) + r$.
A hybrid $n$-path $\sigma$ whose pieces are $r$-shortest is at most $nr$ longer than any path visiting the endpoints of its pieces in the same order.
More generally, for any other path~$\tau$, we bound $\|\sigma\|$ in terms of $n$, $r$, $\|\tau\|$, and $\dF(\sigma,\tau)$ in \cref{lem:lower_bound_lengths}.

\begin{lemma}
\label{lem:lower_bound_lengths}
    Let $r\geq 0$ and let $\sigma$ be a hybrid $n$-path all of whose pieces are $r$-shortest.
    For any hybrid path $\tau$, we have $\|\sigma\| \leq \|\tau\| + (4n-2) \cdot \dF(\sigma, \tau) + nr$.
\end{lemma}
\begin{proof}
    Let $(f, g)$ be a $\delta$-matching between $\sigma$ and $\tau$ for arbitrary $\delta>\dF(\sigma,\tau)$.
    For $i = 1, \dots, n$, let $\sigma[x_{2i-1}, x_{2i}]$ be the pieces of $\sigma$.
    Pick $y_1\leq\dots\leq y_{2n}$ such that $(f,g)$ matches $x_i$ to~$y_i$.
    
    By the triangle inequality we have
    \begin{align*}
        d(\sigma(x_i), \sigma(x_{i+1})) &\leq d(\sigma(x_i), \tau(y_i)) + d(\tau(y_i), \tau(y_{i+1})) + d(\tau(y_{i+1}), \sigma(x_{i+1})) \\
        &\leq 2\delta + d(\tau(y_i), \tau(y_{i+1})).
    \end{align*}
    By our definition of length, $d(\tau(y_i), \tau(y_{i+1})) \leq \|\tau[y_i, y_{i+1}]\|$.
    Using that $y_i \leq y_{i+1}$, we obtain
    \[
        \sum_i d(\sigma(x_i), \sigma(x_{i+1})) \leq (2n-1) \cdot 2\delta + \sum_i \|\tau[y_i, y_{i+1}]\| = (4n-2)\delta + \|\tau\|.
    \]
    By definition of length, we have
    \[
        \|\sigma\| = \sum_{\text{$i$ odd}} \|\sigma[x_i, x_{i+1}]\|  + \sum_{\text{$i$ even}} d(\sigma(x_i), \sigma(x_{i+1})).
    \]
    Because each piece of $\sigma$ is $r$-shortest, we have $\|\sigma[x_i, x_{i+1}]\| \leq d(\sigma(x_i), \sigma(x_{i+1})) + r$ for all odd $i$.
    It therefore follows that
    \[
        \|\sigma\|
        \leq \sum_i d(\sigma(x_i), \sigma(x_{i+1})) + nr
        \leq (4n - 2)\delta + \|\tau\| + nr.
    \]
    This holds for any $\delta>\dF(\sigma, \tau)$, and the claim follows in the limit as $\delta$ approaches $\dF(\sigma,\tau)$.
\end{proof}
In \cref{lem:optimal_is_close_to_diagonal}, we use \cref{lem:lower_bound_lengths} to bound the $\beta$-diagonality of any $\delta$-matching.
\begin{lemma}
\label{lem:optimal_is_close_to_diagonal}
    Let $r\geq 0$.
    Let $\sigma$ and $\tau$ respectively be hybrid $n$- and $m$-paths, all of whose pieces are $r$-shortest.
    Let $N = \max\{n, m\}$.
    Any $\delta$-matching between $\sigma$ and $\tau$ is a $((4N-2)\delta+Nr)$-diagonal matching.
\end{lemma}
\begin{proof}
    If a $\delta$-matching matches $x$ and $y$, then $\dF(\sigma[0,x],\tau[0,y])\leq\delta$.
    So by \cref{lem:lower_bound_lengths},
    \begin{align*}
        \ell_\sigma(x) &\leq \ell_\tau(y) + (4n-2) \cdot \delta + nr\text{, and}\\
        \ell_\tau(y) &\leq \ell_\sigma(x) + (4m-2) \cdot \delta + mr\text{, so}\\
        |\ell_\sigma(x) - \ell_\tau(y)| &\leq (4N-2) \cdot \delta + Nr.
    \end{align*}
    It follows that any $\delta$-matching is a $((4N-2)\delta+Nr)$-diagonal matching.
\end{proof}
Let $\beta(\sigma,\tau)$ be the infimum $\beta$ for which a $\beta$-diagonal matching exists.
Let $L_\sigma$ and $L_\tau$ be the images of $\ell_\sigma$ and $\ell_\tau$, respectively; that is, the sets of prefix-lengths of $\sigma$ and $\tau$.
In \cref{lem:connection_to_Hausdorff}, we show a connection between $\beta(\sigma, \tau)$ and the Hausdorff distance between $L_\sigma$ and $L_\tau$.
The \emph{Hausdorff distance} between sets $A$ and $B$ is defined as
\[
    d_H(A,B) = \max\{\sup_{a\in A}\inf_{b\in B} d(a,b), \sup_{b\in B}\inf_{a\in A} d(a,b)\}.
\]

\begin{lemma}
\label{lem:connection_to_Hausdorff}
    If $L_\sigma$ and $L_\tau$ are closed sets, then $\beta(\sigma, \tau) = d_H(L_\sigma, L_\tau)$.
\end{lemma}
\begin{proof}
    A $\beta$-diagonal matching exists only if $\beta \geq d_H(L_\sigma, L_\tau)$.
    We prove $\beta(\sigma, \tau) \leq d_H(L_\sigma, L_\tau)$.

    Let $L = L_\sigma \cup L_\tau$.
    Because the set $L_\sigma$ is closed, there exists a monotone surjection $f \from L \to L_\sigma$ that maps any $\ell\in L$ to a closest point in $L_\sigma$.
    Similarly, there is a monotone surjection $g \from L \to L_\tau$ that maps any $\ell\in L$ to a closest point in $L_\tau$.
    Since $f(\ell) = \ell$ for $\ell \in L_\sigma$, and $g(\ell) = \ell$ for $\ell \in L_\tau$, we obtain $d_H(L_\sigma, L_\tau) = \max_{\ell \in L} |f(\ell) - g(\ell)|$.
    
    Because $f$ and $g$ are monotone surjections, there is a $\beta$-diagonal matching for $\beta = \max_{\ell \in L} |f(\ell) - g(\ell)|$, and so $\beta(\sigma, \tau) \leq \beta = d_H(L_\sigma, L_\tau)$.
    This completes the proof.
\end{proof}

\noindent
Let $\delta^*_\beta$ be the minimum cost over all $\beta$-diagonal matchings, and let $\delta^*$ be the limit of $\delta^*_\beta$ as $\beta$ tends to $\beta(\sigma,\tau)$.
In \cref{lem:diagonal_cost}, we upper bound $\delta^*$ in terms of $N$, $r$, and $\dF(\sigma,\tau)$.

\begin{lemma}
\label{lem:diagonal_cost}
    Let $r\geq 0$.
    Let $\sigma$ and $\tau$ respectively be hybrid $n$- and $m$-paths, all of whose pieces are $r$-shortest.
    Let $N = \max\{n, m\}$.
    Then $\delta^*\leq (8N-3)\dF(\sigma,\tau)+2Nr$.
\end{lemma}
\begin{proof}
    For some arbitrary $\delta>\dF(\sigma,\tau)$, consider a $\delta$-matching $(f,g)$ between $\sigma$ and $\tau$.
    By \cref{lem:optimal_is_close_to_diagonal}, $(f,g)$ is a $((4N-2)\delta+Nr)$-diagonal matching.
    Therefore $\beta(\sigma,\tau)\leq (4N-2)\delta+Nr$.
    Moreover, by \cref{lem:comparing_diagonal_matchings}, we have $\delta^*_\beta\leq\delta+(4N-2)\delta+Nr+\beta$ whenever a $\beta$-diagonal matching exists.
    Therefore, in the limit as $\beta$ approaches $\beta(\sigma,\tau)$, we obtain
    \begin{align*}
        \delta^* 
        &\leq \delta+(4N-2)\delta+Nr+\beta(\sigma,\tau)\\
        &\leq \delta+(4N-2)\delta+Nr+(4N-2)\delta+Nr\\
        &= (8N-3)\delta+2Nr.
    \end{align*}
    The claim follows in the limit as $\delta$ approaches $\dF(\sigma,\tau)$.
\end{proof}

\subsection{Approximating using the approximate decider}
\label{sub:optimization}

In this section, we for any $\eps>0$, turn a given $\alpha$-approximate decider into an $(\alpha+\eps)$-approximation algorithm for the \f distance.
To do so, we first compute an over-estimate~$\delta^+$ of $\dF(\sigma, \tau)$, and subsequently refine it using the given decider.
We in particular use some $\delta^+$ that is within a factor $\mathrm{poly}(n, m, r / \dF(\sigma, \tau))$ of the \f~distance.
In this section, we assume the pieces of $\sigma$ and $\tau$ to be continuous.

The value $\delta^*$ bounded by \cref{lem:diagonal_cost} is a useful over-estimate.
Instead of computing~$\delta^*$ exactly, we compute a suitable over-estimate $\delta^+$, assuming \cref{prop:ortho-convex} as usual.
We use the oracles of \cref{def:length_oracle,def:distance_oracle}:

\begin{definition}\label{def:length_oracle}
    A \emph{length oracle} for $\sigma$ supports two queries:
    \begin{enumerate}
        \item[1.] Given a breakpoint $x_i$ of $\sigma$, the oracle reports the length $\ell_\sigma(x_i)=\|\sigma[0, x_i]\|$.
        \item[2.] Given a piece $\sigma[x_i, x'_i]$ of $\sigma$ and a length $\ell \in \ell_\sigma([x_i, x'_i])$, it reports some $x$ with $\ell_\sigma(x) = \ell$.
    \end{enumerate}
\end{definition}

\noindent
We assume access to length oracles for $\sigma$ and~$\tau$ with query time $T_L = \Omega(1)$.

\begin{definition}\label{def:distance_oracle}
    A \emph{distance oracle} for $(\sigma,\tau)$ queried with $(x,y)\in X\times Y$ reports $d(\sigma(x), \tau(y))$.
\end{definition}
We assume access to a distance oracle with query time $T_D = \Omega(1)$.

We define our over-estimate $\delta^+$ as follows.
For $i = 1, \dots, n$ and $j = 1, \dots, m$, let $\sigma[x_{2i-1}, x_{2i}]$ and $\tau[y_{2j-1}, y_{2j}]$ be the pieces of $\sigma$ and $\tau$, respectively.
For each breakpoint $x_i$ of $\sigma$, let $\tilde{y}_i$ be the minimum value such that $|\ell_\sigma(x_i) - \ell_\tau(\tilde{y}_i)| \leq |\ell_\sigma(x_i) - \ell_\tau(y)|$ for all $y$.
Define $\tilde{x}_j$ symmetrically for each breakpoint $y_j$ of $\tau$.
We set
\[
    \delta^+ = \max\{ \max_i d(\sigma(x_i), \tau(\tilde{y}_i)), \max_j d(\sigma(\tilde{x}_j), \tau(y_j)) \} + 2\beta(\sigma, \tau).
\]
We will show that $\delta^+$ is a suitable upper bound that can be computed in $O((T_L + T_D) (n+m))$ time.
We first prove that $\delta^+\geq\dF(\sigma, \tau)$ in \cref{lem:over-estimate}, and upper bound $\delta^+$ in \cref{lem:upper_bound_quality}.

\begin{lemma}
\label{lem:over-estimate}
    If $\F_{\delta^+}(\sigma, \tau)$ satisfies \cref{prop:ortho-convex}, then $\delta^+ \geq \dF(\sigma, \tau)$.
\end{lemma}
\begin{proof}
    Since the pieces of $\sigma$ and $\tau$ are continuous, there exists a $\beta(\sigma, \tau)$-diagonal matching.
    Let $\pi$ be the path through $\D(\sigma, \tau)$ corresponding to such a matching.
    Consider a partition of $\pi$ into subpaths whose endpoints lie on the edges of the parameter space grid.
    Let $\pi'$ be a subpath in this partition, and let $p$ and $q$ be its endpoints.
    We argue that $p$ is $\delta^+$-free.

    Suppose without loss of generality that $p$ lies on the left side of a cell.
    That is, $p = (x_i, y)$ for some breakpoint $x_i$ of $\sigma$ and some $y\in Y$.
    By $\beta(\sigma,\tau)$-diagonality and our choice of~$\tau(\tilde{y}_i)$,
    \[
        | \ell_\sigma(x_i) - \ell_\tau(\tilde{y}_i) | \leq | \ell_\sigma(x_i) - \ell_\tau(y) | \leq \beta(\sigma, \tau).
    \]
    Thus, the length of the hybrid subpath $\tau[\tilde{y}_i,y]$ is at most $2\beta(\sigma, \tau)$.
    We obtain that
    \begin{align*}
        d(\sigma(x_i), \tau(y)) &\leq d(\sigma(x_i), \tau(\tilde{y}_i)) + d(\tau(\tilde{y}_i), \tau(y))\\
        &\leq d(\sigma(x_i), \tau(\tilde{y}_i)) + 2\beta(\sigma, \tau)\\
        &\leq \delta^+.
    \end{align*}
    Hence, $p$ is $\delta^+$-free.
    It follows from a symmetric argument that $q$ is $\delta^+$-free.
    
    Although the endpoints of $\pi'$ are $\delta^+$-free, other points of $\pi'$ may not.
    Still, because $\F_{\delta^+}(\sigma, \tau)$ satisfies \cref{prop:ortho-convex}, $q$ is $\delta^+$-reachable from $p$.
    Thus, $\pi'$ can be replaced by a bimonotone path from $p$ to $q$ in $\F_{\delta^+}(\sigma, \tau)$.
    Replacing all subpaths in the partition of $\pi$ in this manner, we obtain a $\delta^+$-matching between $\sigma$ and $\tau$.
\end{proof}

\begin{lemma}
\label{lem:upper_bound_quality}
    $\delta^+ = O((n+m) (\dF(\sigma, \tau) + r))$.
\end{lemma}
\begin{proof}
    Because the functions $\ell_\sigma$ and $\ell_\tau$ are non-decreasing, there exists a bimonotone path in $\D(\sigma, \tau)$ through all points $(x_i, \tilde{y}_i)$ and $(\tilde{x}_j, y_j)$.
    Additionally, the points $(x_1, \tilde{y}_1)$ and $(\tilde{x}_j, y_j)$ are both the bottom-left corner of $\D(\sigma, \tau)$, and at least one of $(x_{2n}, \tilde{y}_{2n})$ and $(\tilde{x}_{2m}, y_{2m})$ is the top-right corner of $\D(\sigma, \tau)$.
    Thus, there exists a matching through all points $(x_i, \tilde{y}_i)$ and~$(\tilde{x}_j, y_j)$.
    In particular, there exists a $\beta(\sigma, \tau)$-matching through these points.
    Such a matching has cost $O((n+m) (\dF(\sigma, \tau) + r))$ by \cref{lem:diagonal_cost}, proving that all points $(x_i, \tilde{y}_i)$ and $(\tilde{x}_j, y_j)$ are $O((n+m) (\dF(\sigma, \tau) + r))$-free.
    Hence, $\delta^+ = O((n+m) (\dF(\sigma, \tau) + r))$.
\end{proof}
With the bounds on $\delta^+$ proven, we proceed to give an efficient algorithm for its computation:

\begin{lemma}
\label{lem:computing_over-estimate}
    We can compute $\delta^+$ in $O((T_L + T_D) (n+m))$ time.
\end{lemma}
\begin{proof}
    We first compute the value of $\beta(\sigma, \tau)$.
    For this, we construct the images $L_\sigma$ and $L_\tau$ of the functions $\ell_\sigma$ and $\ell_\tau$, respectively.
    Because the pieces of $\sigma$ and $\tau$ are continuous, it suffices to compute the images of $\ell_\sigma$ and $\ell_\tau$ at the breakpoints.
    This takes $O(T_L \cdot (n+m))$ time with the length oracle.
    By piecewise-continuity of $\sigma$ and $\tau$, $L_\sigma$ and $L_\tau$ are unions of at most $n$ respectively $m$ sorted closed intervals.
    By \cref{lem:connection_to_Hausdorff}, $\beta(\sigma, \tau) = d_H(L_\sigma, L_\tau)$, which we can compute in $O(n+m)$ time~\cite{rino26STL_fragment}.

    Next we compute the values $d(\sigma(x_i), \tau(\tilde{y}_i))$ for all breakpoints $x_i$ of $\sigma$.
    Given the lengths of the prefixes up to breakpoints, we can determine the piece of $\tau$ containing $\tau(\tilde{y}_i)$, for all $i$, in $O(n+m)$ time with a simultaneous scan of $\sigma$ and $\tau$.
    For a breakpoint $x_i$ of $\sigma$, we compute the value $d(\sigma(x_i), \tau(\tilde{y}_i))$ as follows.
    Let $\tau[y_j, y_{j+1}]$ be the piece containing $\tau(\tilde{y}_i)$.
    Set
    \[
        \ell = \begin{cases}
            \ell_\sigma(x_i) & \text{if $\ell_\sigma(x_i) \in [\ell_\tau(y_j), \ell_\tau(y_{j+1})]$,}\\
            \ell_\tau(y_j) & \text{if $\ell_\sigma(x_i) > \ell_\tau(y_j)$,}\\
            \ell_\tau(y_{j+1}) & \text{otherwise.}
        \end{cases}
    \]
    We have $\ell_\tau(\tilde{y}_i) = \ell$.
    We query the length oracle for a point $\tau(y)$ with $\ell_\tau(y) = \ell_\tau(\tilde{y}_i)$.
    By definition, the length of the hybrid subpath of $\tau$ between $\tau(y)$ and $\tau(\tilde{y}_i)$ is zero.
    So, by the definition of a metric space, the points $\tau(y)$ and $\tau(\tilde{y}_i)$ coincide.
    Hence, $d(\sigma(x_i), \tau(\tilde{y}_i)) = d(\sigma(x_i), \tau(y))$, and we can compute its value with a query to the distance oracle.

    The above procedure takes $O(T_L + T_D)$ time.
    We can compute $d(\sigma(\tilde{x}_j), \tau(y_j))$ for any breakpoint $y_j$ of $\tau$ in $O(T_L + T_D)$ time in a symmetric manner.
    Doing so for all breakpoints of $\sigma$ and $\tau$ thus takes $O((T_L + T_D) (n+m))$ time.
    This dominates the time taken to compute the lengths of prefixes up to the breakpoints, completing the proof.
\end{proof}
We proceed to approximate $\dF(\sigma, \tau)$, using the given $\alpha$-approximate decider.
Let $\delta_i = \delta^+ / 2^i$ for integers $i \geq 1$.
We find a value $\delta_i$ for which the decider reports that $\delta_i < \dF(\sigma, \tau)$, but reports that $\alpha\delta_{i-1} \geq \dF(\sigma, \tau)$.
Through exponential search, we can find such a value using $O(\log i)$ calls to the decider.
We follow this up with a binary search over $O(\log_{1+\eps} \alpha)$ values in the interval $[\delta_i, \alpha\delta_{i-1}]$, namely the values $(1+\eps)^j \delta_i$ for $j = 1, \dots, O(\log_{1+\eps} \alpha)$.
Again using the decider during this second search, we obtain an $(\alpha+\eps)$-approximation of $\dF(\sigma, \tau)$ using an additional $O(\log \log_{1+\eps} \alpha) = O(\log (\frac{1}{\eps} \log \alpha))$ calls to the decider.

We have $\delta_i = \delta^+ / 2^i = O((n+m) \cdot (\dF(\sigma, \tau) + r) / 2^i)$ by \cref{lem:upper_bound_quality}.
We also have $\delta_i = \Theta(\dF(\sigma, \tau))$ by the above.
Hence, $i = O(\log (n+m + r / \dF(\sigma, \tau)))$.
Together with the final binary search, our algorithm performs $O(\log (\frac{i}{\eps} \log \alpha))$ calls to the approximate decider.
Factoring in the $O((T_L + T_D) \cdot (n+m))$ time taken to compute $\delta^+$ with \cref{lem:computing_over-estimate}, we obtain:

\begin{theorem}
    Let $\sigma$ and $\tau$ respectively be hybrid $n$- and $m$-paths with $m \leq n$.
    Suppose
    \begin{itemize}
        \item the pieces of $\sigma$ and $\tau$ are continuous and $r$-shortest for some $r \geq 0$, and
        \item $\F_\delta(\sigma, \tau)$ satisfies \cref{prop:ortho-convex} for all $\delta \geq 0$.
    \end{itemize}
    Suppose we are given an $\alpha$-approximate decider that runs in time $T(n)$ and uses space $S(n)$.
    For any $\eps > 0$, we can compute an $(\alpha+\eps)$-approximation of $\dF(\sigma, \tau)$ in
    \[
        O(T(n) \cdot \log \left( \frac{1}{\eps} \log \left(n + \frac{r}{\dF(\sigma, \tau)} \right) \log \alpha \right) + (T_L + T_D) \cdot n) \text{ time and } O(S(n) + n) \text{ space.}
    \]
\end{theorem}
With our $3$-approximate decider from \cref{thm:approx_decider}, we obtain the following:

\begin{corollary}
    Let $\sigma$ and $\tau$ respectively be hybrid $n$- and $m$-paths with $m \leq n$.
    Suppose
    \begin{itemize}
        \item the pieces of $\sigma$ and $\tau$ are continuous and $r$-shortest for some $r \geq 0$, and
        \item for every $\delta \geq 0$, all of $\F_\delta(\sigma, \tau)$, $\F_\delta(\sigma, \sigma)$, and $\F_\delta(\tau, \tau)$ satisfy \cref{prop:ortho-convex}.
    \end{itemize}
    For any $\eps > 0$, we can compute a $(3+\eps)$-approximation of $\dF(\sigma, \tau)$ in
    \[
        O{\left( (nm^{2/3} \cdot (T_O^3 + T_O^2 \log^2 n + T_O \log^3 n)^{1/3} \cdot \log \left( \frac{1}{\eps} \log \left( n + \frac{r}{\dF(\sigma, \tau)} \right) \right) + (T_L + T_D) \cdot n \right)}
    \]
    time and $O(nm^{2/3} \cdot (T_O^3 + T_O^2 \log^2 n + T_O \log^3 n)^{1/3})$ space.
\end{corollary}

\subsection{Approximation results for specific metric spaces}
\label{sub:specific_spaces}

Our approximation algorithm is very general, assuming nothing about the ambient metric space, and only some natural assumptions on the inputs.
In this section, we consider some metric spaces used commonly in the context of the \f distance, and state what our approximation algorithm implies for those.

Define a \emph{hybrid polyline} to be a hybrid path in $\R^d$, all of whose pieces are segments.
In this section, we consider hybrid polylines in subspaces of $\R^d$.
Our algorithm can efficiently approximate both the continuous and the discrete \f distance between hybrid polylines, under some commonly studied metrics.
Let $\sigma$ and $\tau$ be two polylines with $n$ and $m \leq n$ vertices, respectively.
For the continuous \f distance, the edges of $\sigma$ and $\tau$ are already parameterized over interior-disjoint intervals, and so $\sigma$ and $\tau$ are hybrid $n$- and $m$-paths, respectively.
For the discrete \f distance, we represent discrete paths $\sigma$ and $\tau$ as hybrid $n$- and $m$-paths using parameterizations consisting of a singleton interval per vertex.

First, let $\sigma$ and $\tau$ be hybrid polylines in $\R^d$, with as underlying metric the $L_p$ norm for some $p \geq 1$.
The intersection between $\F_\delta(\sigma, \tau)$ and a cell of the parameter space is convex for all $\delta \geq 0$~\cite{alt95continuous_frechet}, so \cref{prop:ortho-convex} is satisfied.
We can trivially build a $\delta$-free space oracle and distance oracle with query times $T_O = O(d)$ and $T_D = O(d)$.
As a polyline is a piecewise-linear function, we can preprocess it in $O(nd)$ time to build a length oracle with constant query time.
Since every edge of a polyline is a continuous $0$-shortest path, we obtain:

\begin{corollary}
\label{cor:result_real_space}
    For $p \geq 1$, let $\sigma$ and $\tau$ be hybrid polylines in $(\R^d, L_p)$ with $n$ and $m\leq n$ vertices respectively.
    For any $\eps > 0$, we can $(3+\eps)$-approximate the \f distance between $\sigma$ and $\tau$ in $O{\mleft( nm^{2/3} \cdot (d^3 + d^2 \log^2 n + d \log^3 n)^{1/3} \cdot \log \mleft( \frac{1}{\eps} \log n \mright) \mright)}$ time.
\end{corollary}
Next, suppose that $\sigma$ and $\tau$ lie in a simple $k$-gon $P \subseteq \R^2$ endowed with the geodesic $L_2$ metric.
The intersection between $\F_\delta(\sigma, \tau)$ and a cell of the free space is not necessarily convex, but it is connected and ortho-convex~\cite{cook10geodesic_frechet}, so \cref{prop:ortho-convex} is satisfied.
In this setting, we can in $O(k)$ time build a $\delta$-free space oracle and distance oracle with query times $T_O = O(\log k)$ and $T_D = O(\log k)$~\cite{cook10geodesic_frechet}.
For the length oracle in the continuous setting, the same oracle can be used as for $(\R^d, L_2)$.
In the discrete setting, we require $O(n \log k)$ additional preprocessing instead, to compute the distances between consecutive vertices.
Every edge is still a continuous $0$-shortest path, and so we obtain:

\begin{corollary}\label{cor:result_geodesic_Euclidean}
    Let $P$ be a simple $k$-gon endowed with the geodesic Euclidean metric.
    Let~$\sigma$ and $\tau$ be hybrid polylines in $P$, respectively with $n$ and $m\leq n$ vertices.
    For any $\eps > 0$, we can compute a $(3+\eps)$-approximation of the \f distance between $\sigma$ and $\tau$ in $O{\mleft( nm^{2/3} \cdot \log^{4/3} n \cdot \log \mleft( \frac{1}{\eps} \log n \mright) + k\mright)}$ time.
\end{corollary}
Lastly, we consider the setting where $\sigma$ and $\tau$ lie on the real line $\R$ endowed with the standard norm.
We can improve upon \cref{cor:result_real_space} in this case, to obtain a $3$-approximation in roughly the same time bound.
In this setting, using the arguments from~\cite{alt95continuous_frechet}, the \f distance between~$\sigma$ and $\tau$ is either the distance between a vertex of~$\sigma$ and a vertex of $\tau$, or half the distance between two vertices of~$\sigma$ or $\tau$.
We can represent the set of all these distances as the Cartesian sum $U \oplus V = \{u+v \mid u \in U, v \in V\}$ of two sets $U$ and $V$ of $O(n+m)$ values, namely by choosing $U=V=\{s\cdot u\mid s\in\{-1,-\frac{1}{2},\frac{1}{2},1\} \text{ and $u$ is a vertex of $\sigma$ or $\tau$}\}$.
Selection in $U \oplus V$ can be done in $O(n+m)$ time~\cite{mirzaian85selection_XY}, so we can binary search over these values with $O(n+m)$ overhead per step.
The smallest value $\delta$ for which our decider reports $\dF(\sigma, \tau) \leq 3\delta$ is then a $3$-approximation, as $\delta \leq \dF(\sigma, \tau)$.
Thus, we obtain:

\begin{corollary}
    \label{cor:result_one_dimensional}
    Let $\sigma$ and $\tau$ be hybrid polylines in $(\R, |\cdot|)$, respectively with $n$ and $m\leq n$ vertices.
    We can compute a $3$-approximation of the \f distance between $\sigma$ and $\tau$ in $O(nm^{2/3} \log^2 n)$ time.
\end{corollary}

\section{Discussion}\label{sec:discussion}

    Our approximate decider resembles that by Cheng, Huang, and Zhang \cite{cheng25constant_factor_frechet}.
    Their algorithm takes as input two polylines $\sigma$ and $\tau$ in $\R^d$.
    They also partition $\sigma$ and $\tau$ into smaller subcurves $\sigma_i$ and $\tau_j$, and build data structures for each.
    Crucially, our data structures are significantly more time- and space-efficient to build and query.
    Ultimately, this means that we do not have to subdivide $\sigma$ and $\tau$ as much, resulting in fewer pairs $(i,j)$ over which to propagate reachability.
    To propagate reachability for a pair $(i,j)$, the algorithm of~\cite{cheng25constant_factor_frechet} computes something akin to our surrogates, but theirs are less accurate and rely on randomization, resulting in an expected rather than worst-case strongly subquadratic running time.
    Finally, the data structures in \cite{cheng25constant_factor_frechet} must compensate for the less accurate surrogates, resulting in a worse approximation factor compared to ours.

    Although the approximation factor of our approximate decider is optimal under SETH, it is unclear whether faster constant-factor approximate deciders exist.
    Let $\xi_\alpha$ be the infimum value $\xi$ for which an $\alpha$-approximate decider with running time $O(n^\xi)$ exists.
    Then, assuming the strong exponential time hypothesis, we have $\xi_\alpha=2$ for all $\alpha<3$.
    Our results imply that $\xi_\alpha\leq 5/3$ for all $\alpha \geq 3$.
    We wish to still gain a better understanding of $\xi_\alpha$ as a function of $\alpha$.
    
    It is also unclear whether the time and space tradeoffs of \cref{sec:time_and_space} can be improved.

\bibliographystyle{plainurl}
\bibliography{bibliography}

\appendix

\newpage
\section{An exact decider for hybrid paths}
\label{app:exact_decision}

The classical decision algorithm by Alt and Godau~\cite{alt95continuous_frechet} determines, for polylines $\sigma$ and $\tau$ in $(\R^d, L_p)$, whether $\dF(\sigma, \tau) \leq \delta$ for a given value $\delta \geq 0$.
We show in \cref{lem:alt_godau_extended} that their algorithm easily extends to hybrid paths $\sigma$ and $\tau$ in any metric space, provided that $\F_\delta(\sigma, \tau)$ satisfies \cref{prop:ortho-convex}, and that we have access to a free space oracle.

\begin{lemma}
\label{lem:alt_godau_extended}
    Let $\sigma$ and $\tau$ be hybrid $n$- and $m$-paths, respectively, and let $\delta \geq 0$ be a value for which $\F_\delta(\sigma, \tau)$ satisfies \cref{prop:ortho-convex}.
    We can determine whether $\dF(\sigma, \tau) \leq \delta$ in $O(T_O \cdot nm)$ time and $O(n+m)$ space.
\end{lemma}
\begin{proof}
    Let $\sigma$ and $\tau$ respectively be parameterized over sets $X$ and $Y$. 
    Let their pieces be $\sigma[x_1, x'_1], \dots, \sigma[x_n, x'_n]$ and $\tau[y_1, y'_1], \dots, \tau[y_m, y'_m]$.
    For our algorithm, we determine whether $(x_1, y_1)$ can $\delta$-reach $(x'_n, y'_m)$.
    For simplicity of exposition, we say that a point $(x, y)$ is \emph{$\delta$-reachable} if it is $\delta$-reachable from $(x_1, y_1)$.

    We extend the standard decision algorithm by Alt and Godau~\cite{alt95continuous_frechet} to work for hybrid paths in a metric space.
    Let $C_{i, j}$ denote the cell $[x_i, x'_i] \times [y_j, y'_j]$ of the parameter space.
    We process the cells of $\D(\sigma, \tau)$ in some topological order, so that cell $C_{i, j}$ is processed after $C_{i-1, j}$ and $C_{i, j-1}$ (if they exist) have been processed.
    For each cell $C_{i, j}$, we compute the set of $\delta$-reachable points on the boundary of $C_{i, j}$.
    Because $\F_{\delta}(\sigma, \tau)$ satisfies \cref{prop:ortho-convex}, the set of $\delta$-reachable points on any side of $C_{i, j}$ forms a line segment.

    We process cell $C_{i, j}$ in $O(T_O)$ time as follows.
    We compute the $\delta$-free points on the boundary of $C_{i, j}$ with four queries to the $\delta$-free space oracle, taking $O(T_O)$ time.
    We show how to compute the $\delta$-reachable points on the bottom side of $C_{i, j}$; the other side is handled symmetrically.
    Let $(x, y)$ be a point on the bottom side of $C_{i, j}$.
    We distinguish between the case that $x = x_i$, i.e., that the point is the bottom-left corner of $C_{i, j}$, and the case~$x > x_i$.

    If $x = x_i$, then the point $(x, y)$ is $\delta$-reachable if and only if $(x'_{i-1}, y_j)$, $(x_i, y'_{j-1})$, or $(x'_{i-1}, y'_{j-1})$ is $\delta$-reachable.
    We can determine whether this condition holds in $O(1)$ time.
    Suppose $x > x_i$.
    The point $(x, y)$ is $\delta$-reachable if and only if the segment $[x-\eps, x] \times \{y'_{j-1}\}$ on the top side of $C_{i, j-1}$ contains a point that is $\delta$-reachable, for all $\eps > 0$.
    We can in $O(1)$ time determine the set of points on the bottom side of $C_{i, j}$ that meet this condition.
    
    By the above, we can compute the set of $\delta$-reachable points on the bottom and left sides of $C$ in $O(1)$ time.
    Because $C$ has the $\delta$-free-reachability property, we can then compute the set of $\delta$-reachable points on the top and right sides of $C$ in $O(1)$ time.

    To summarize, we process each cell in $O(T_O)$ time.
    Once the last cell is processed, we know whether $(x'_n, y'_m)$ is $\delta$-reachable from $(x_1, y_1)$ or not, and can answer the decision question.
    Our decision algorithm thus runs in $O(T_O \cdot nm)$ time.
    We require knowledge of the $\delta$-reachable points on the boundary of a cell $C$ only when the cells directly above and to the right of $C$ have not been processed yet.
    Thus, we need to keep these points in memory for only $O(n+m)$ cells, requiring $O(n+m)$ space.
\end{proof}

\section{A faster approximate decider for the imbalanced case}
\label{sec:decision_imbalanced}

Our surrogate computation algorithm from \cref{sec:surrogate} implies a simple $3$-approximate decider for the imbalanced case.
Given $\delta \geq 0$ and hybrid $n$- and $m$-paths $\sigma$ and $\tau$ with $m \leq n$, it reports either that $\dF(\sigma, \tau) \leq 3\delta$, or that $\dF(\sigma, \tau) > \delta$.
When $\dF(\sigma, \tau) \in (\delta, 3\delta]$, it may report either answer.
We assume that $\F_\delta(\sigma, \tau)$ and $\F_{2\delta}(\tau, \tau)$ satisfy \cref{prop:ortho-convex}.

Our algorithm starts similar to the recent $3$-approximate decision algorithm by Blank~\cite{blank26imbalanced_frechet}, which assumes $T_O = O(1)$.
Both our algorithm and Blank's run in $O(T_O \cdot (n+m^2))$ time, and both set out to compute a simplification of $\tau$ of complexity $O(m)$.
The way Blank's algorithm constructs the simplification is essentially a special case (where $k=m$) of our algorithm for computing a $(\tau,m,\delta)$-surrogate of~$\sigma$ via \cref{thm:surrogate}.
This special case effectively boils down to a simultaneous linear scan over both input paths.
By \cref{thm:surrogate}, the algorithm will report a surrogate for the entirety of $\sigma$ if $\dF(\sigma, \tau) \leq \delta$.
Thus, should the algorithm not provide such a surrogate, our algorithm (just like Blank's) immediately reports $\dF(\sigma, \tau) > \delta$.

Let $\calS$ be the computed $(\tau, m, \delta)$-surrogate.
By \cref{thm:surrogate}, this is a hybrid $O(m)$-path with $\dF(\sigma, \calS) \leq \delta$.
Since the \f distance satisfies the triangle inequality, if $\dF(\sigma, \tau) \leq \delta$, then $\dF(\calS, \tau) \leq 2\delta$, and conversely if $\dF(\calS, \tau) \leq 2\delta$, then $\dF(\sigma, \tau) \leq 3\delta$.
Thus, it remains to determine whether $\dF(\calS, \tau) \leq 2\delta$.
Crucially, the surrogate is of such low complexity that quadratic time algorithms can be used to decide its distance to $\tau$.

The main difference between the two algorithms is that, for the continuous Fr\'echet distance, Blank converts resulting surrogate (which is a hybrid path) into a continuous path by connecting the endpoints of each ``jump'' by a shortest path.
Although this is a valid operation in $L_p$-norms, this is not always possible in general metric spaces.
The benefit of converting it to a continuous path is that quadratic-time decision algorithms are readily available (provided the necessary distance oracles).
In contrast, we use a generalization to hybrid paths, such as \cref{lem:alt_godau_extended}.

Computing the surrogate $\calS$ takes $O(T_O \cdot n)$ time and $O(n)$ space.
Using \cref{lem:alt_godau_extended} we can determine whether $\dF(\calS, \tau) \leq 2\delta$ in $O(T_O \cdot m^2)$ time and $O(m)$ space, so we obtain:

\begin{theorem}\label{thm:decider_imbalanced}
    There is an $O(T_O \cdot (n + m^2))$ time, $O(n+m)$ space, $3$-approximate decider.
\end{theorem}
Instead of deciding whether $\dF(\calS, \tau) \leq 2\delta$ in an exact manner, we can instead use our main $3$-approximate decider of \cref{thm:approx_decider}.
This algorithm then reports either that $\dF(\calS, \tau) \leq 6\delta$, or $\dF(\calS, \tau) > 2\delta$.
In the former case, we have $\dF(\sigma, \tau) \leq 7\delta$, and in the latter, we have $\dF(\sigma, \tau) > \delta$.
Thus, this results in a $7$-approximate decider.

In this context, our decision algorithm requires that $\F_{2\delta}(\calS, \tau)$, $\F_{4\delta}(\calS, \calS)$, and $\F_{4\delta}(\tau, \tau)$ all satisfy \cref{prop:ortho-convex}.
Since $\calS$ is formed by a sequence of hybrid subpaths of $\tau$, the second requirement is met when the third is.
Therefore, we obtain:

\begin{corollary}\label{cor:decider_imbalanced_7apx}
    Let $\sigma$ and $\tau$ be hybrid $n$- and $m$-paths, respectively, with $m \leq n$.
    Let $\delta \geq 0$ be a parameter for which $\F_\delta(\sigma, \tau)$, $\F_{2\delta}(\tau, \tau)$, and $\F_{4\delta}(\tau, \tau)$ all satisfy \cref{prop:ortho-convex}.
    We can decide whether $\dF(\sigma, \tau) \leq 7\delta$ or $\dF(\sigma, \tau) > \delta$ in $O(T_O \cdot n + m^{5/3} \cdot (T_O^3 + T_O^2 \log^2 m + T_O \log^3 m)^{1/3})=O^*(n + m^{5/3})$ time and $O(T_O \cdot n + m^{5/3} \cdot (T_O^3 + T_O^2 \log^2 m + T_O \log^3 m)^{1/3})=O^*(n + m^{5/3})$ space, where $O^*$ hides polynomials in $T_O$ and $\log m$.
\end{corollary}

\end{document}